\documentclass[11pt]{article}
\usepackage{liyang}
 
\usepackage{tikz}
\usepackage{verbatim}
\usepackage{cancel}

\newcommand{\cnote}[1]{\footnote{{\bf \color{cyan}Chin}: {#1}}}
\newcommand{\rnote}[1]{\footnote{{\bf \color{red}Rocco}: {#1}}}
\newcommand{\snote}[1]{\footnote{{\bf \color{blue}Sandip}: {#1}}}
\newcommand{\xnote}[1]{\footnote{{\bf \color{red}Xi}: {#1}}}

\usepackage{todonotes}

\makeatletter
\newtheorem*{rep@theorem}{\rep@title}
\newcommand{\newreptheorem}[2]{
\newenvironment{rep#1}[1]{
 \def\rep@title{#2 \ref{##1}}
 \begin{rep@theorem}\itshape}
 {\end{rep@theorem}}}
\makeatother
\theoremstyle{plain}

\newcommand{\deslen}{C}

\makeatletter
\newenvironment{proofof}[1]{\par
  \pushQED{\qed}%
  \normalfont \topsep6\p@\@plus6\p@\relax
  \trivlist
  \item[\hskip\labelsep
\emph{    Proof of #1\@addpunct{.}}]\ignorespaces
}{%
  \popQED\endtrivlist\@endpefalse
}
\makeatother

\newcommand{\ignore}[1]{}

\def\colorful{1}

\ifnum\colorful=1

\newcommand{\blue}[1]{{{\color{blue}#1}}}

\newcommand{\gray}[1]{{\color{gray}{#1}}}

\fi
\ifnum\colorful=0

\newcommand{\blue}[1]{{{#1}}}

\newcommand{\gray}[1]{{{#1}}}

\fi

\usepackage{boxedminipage}

\newreptheorem{theorem}{Theorem}
\newtheorem*{theorem*}{Theorem}
\newreptheorem{lemma}{Lemma}
\newreptheorem{proposition}{Proposition}
\newreptheorem{corollary}{Corollary}
\newtheorem*{noclaim*}{Claim}

\newcommand{\Bin}{\mathrm{Bin}}

\newcommand{\overallsc}{N}

\newcommand{\Del}{\mathrm{Del}}

\def\BMA{\textup{{\tt BMA}}\xspace}
\def\DET{\textup{{\tt FindEnd}}\xspace}
\def\CE{\textup{{\tt Coarse-Estimate}}\xspace}
\def\SE{\textup{{\tt Sharp-Estimate}}\xspace}

\def\Align{\textup{{\tt Align}}\xspace}

\newcommand{\dend}{\textup{\textsf{end}}}
\newcommand{\current}{\textnormal{current}\xspace}
\newcommand{\position}{\textup{\textsf{position}}\xspace}
\newcommand{\distance}{\textup{\textsf{distance}}\xspace}

\def\pos{\textup{\textnormal{last}}\xspace}
\def\reach{\textup{\textsf{reach}}\xspace}
\def\tail{\textup{\textsf{tail}}\xspace}
\def\sig{\textup{\textsf{sig}}\xspace}
\def\signature{\sig}
\def\sigg{\sig}
\def\Cyc{\textup{\textnormal{Cyc}}_s\xspace}
\def\nil{\textup{\textnormal{nil}}\xspace}
\usepackage{dsfont}

\begin{document}

\title{Polynomial-time trace reconstruction\\ 
in the low deletion rate regime\vspace{0.2cm}}


\author{
Xi Chen\thanks{Supported by NSF grants CCF-1703925 and IIS-1838154.} 
\\
Columbia University\\
xichen@cs.columbia.edu
\and
Anindya De\thanks{Supported by NSF grants CCF-1926872 and CCF-1910534.}
\\
University of Pennsylvania\\
anindyad@cis.upenn.edu
\and
Chin Ho Lee\thanks{Supported by a grant from the Croucher Foundation and by the Simons Collaboration on Algorithms and Geometry.}
\\
Columbia University\\
c.h.lee@columbia.edu
\and
Rocco A. Servedio\thanks{Supported by NSF grants CCF-1814873, IIS-1838154, CCF-1563155, and by the Simons Collaboration on Algorithms and Geometry.} 
\\
Columbia University\\
rocco@cs.columbia.edu
\and
Sandip Sinha\thanks{Supported by NSF grants  CCF-1714818, CCF-1822809, IIS-1838154, CCF-1617955, CCF-1740833, and by the Simons Collaboration on Algorithms and Geometry.}
\\
Columbia University\\
sandip@cs.columbia.edu
}

\date{}
\maketitle

\thispagestyle{empty}


\begin{abstract} 

In the \emph{trace reconstruction problem}, an unknown source string $x \in \{0,1\}^n$ is transmitted through a probabilistic \emph{deletion channel} which independently deletes each bit with some fixed probability $\delta$ and concatenates the surviving bits, resulting in a \emph{trace} of $x$. The problem is to reconstruct $x$ given access to independent traces. 

Trace reconstruction of arbitrary (worst-case) strings is a challenging problem, with the current state of the art for $\poly(n)$-time algorithms being the 2004 algorithm of Batu et al.~\cite{BKKM04}. This algorithm can reconstruct an arbitrary source string $x \in \{0,1\}^n$ in $\poly(n)$ time provided that the deletion rate $\delta$ satisfies $\delta \leq n^{-(1/2 + \varepsilon)}$ for some $\varepsilon > 0$.

In this work we improve on the result of \cite{BKKM04} by giving a $\poly(n)$-time algorithm for trace reconstruction for any deletion rate $\delta \leq n^{-(1/3 + \varepsilon)}$. Our algorithm works by alternating an alignment-based procedure, which we show effectively reconstructs portions of the source string that are not ``highly repetitive'', with a novel procedure that efficiently determines the length of highly repetitive subwords of the source string.

\end{abstract}

\newpage

\newpage

\setcounter{page}{1}


\section{Introduction} \label{sec:intro}

The \emph{trace reconstruction} problem was proposed almost twenty years ago in works of \cite{Lev01b, Lev01a, BKKM04}, though some earlier variants of the problem were already considered in the 1970s  \cite{Kalashnik73}. 
This problem deals with the \emph{deletion channel}, which works as follows: when an $n$-bit string (the source string) is passed through a deletion channel of rate $\delta$, each coordinate is independently deleted with probability $\delta$. 
The surviving $n' \leq n$ coordinates are concatenated to form the output of the channel, which is referred to as a \emph{trace} of the original source string; we write ``$\bz \sim \Del_\delta(x)$'' to indicate that $\bz$ is a trace generated from source string $x$ according to this probabilistic process.
As discussed in \cite{Mit09}, this channel provides an elegant formalization for the theoretical study of problems involving synchronization errors.

In the \emph{trace reconstruction problem}, independent traces are generated from an unknown and arbitrary source string $x \in \{0,1\}^n$, and the task of the algorithm is to reconstruct (with high probability) $x$ from its traces. 
The trace reconstruction problem is motivated by applications in several domains, including sensor networks and biology \cite{Mit09,ADHR12,DNAS,organick2018random}.
It is also attractive because it is a clean and natural ``first problem'' which already seems to capture much of the difficulty of dealing with the deletion channel.

The problem of trace reconstruction for an arbitrary (worst-case) source string $x$ has proved to be quite challenging.\footnote{We note that the average-case problem, in which the reconstruction algorithm is only required to succeed for a $1-o(1)$ fraction of all possible source strings in $\zo^n$, is much more tractable, with the current state of the art \cite{HPP18,HPPZ20} being an algorithm that uses $\exp(O(\log^{1/3} n))$ traces and runs in $\poly(n)$ time for any deletion rate $\delta$ that is bounded away from 1.}  \cite{BKKM04} gave an algorithm that runs in $\poly(n)$ time, uses $\poly(n)$ traces, and with high probability reconstructs an arbitrary source string $x \in \zo^n$ provided that the deletion rate $\delta$ is at most $n^{-(1/2 + \eps)}$ for some constant $\eps > 0$. 
Unfortunately, the trace reconstruction problem seems to quickly become intractable at higher deletion rates. Holenstein et al.~\cite{HMPW08} gave an algorithm  that runs in time $\exp(O(n^{1/2}))$ and uses $\exp(O(n^{1/2}))$ traces for any deletion rate $\delta$ that is bounded away from 1 by a constant, and this result was subsequently improved in simultaneous and independent works by \cite{DOS17,NazarovPeres17}, both of which gave algorithms with time and sample complexity $\exp(O(n^{1/3})).$  On the lower bounds side,  for $\delta = \Theta(1)$ successively stronger lower bounds on the required sample complexity were given by \cite{MPV14} and \cite{HL18}, with the current state of the art being a $\smash{\tilde{\Omega}(n^{3/2})}$ lower bound due to Chase \cite{Chase19}.

\medskip

{\bf The low deletion rate regime.} The positive result of \cite{DOS17} actually gives an algorithm that is faster than $\exp(O(n^{1/3}))$ if the deletion rate is sufficiently low: \cite{DOS17} shows that for $O(\log^3 n)/n \leq \delta \leq 1/2$, their algorithm runs in time $\exp(O(\delta n)^{1/3})$.  Consequently, for the specific deletion rate $\delta = n^{-(1/2 + \eps)}$, the \cite{DOS17} algorithm runs in time essentially $\exp(O(n^{1/6}))$, and \cite{DOS17} shows that no faster running time or better sample complexity is possible for any ``mean-based'' algorithm, a class of algorithms which includes those of \cite{DOS17,NazarovPeres17,HMPW08}.

Algorithmic approaches other than mean-based algorithms can provably do better at low deletion rates.
This is witnessed by the algorithm of Batu et al.~\cite{BKKM04} which, as described earlier, runs in $\poly(n)$ time and uses $\poly(n)$ samples at deletion rate $\delta = n^{-(1/2 + \eps)}$. 
The main algorithmic component of \cite{BKKM04} is a ``\emph{Bitwise Majority Alignment}'' (\BMA for short) procedure, which is further augmented with a simple procedure to determine the length of long ``runs'' (subwords of $x$ of the form $0^\ell$ or $1^\ell$ with $\ell\ge \sqrt{n}$).  
Roughly speaking, the \BMA algorithm maintains a pointer in each trace and increments those pointers in successive time steps, attempting to always keep almost all of the pointers correctly aligned together. 
The analysis of \cite{BKKM04} shows that the \BMA algorithm succeeds if the source string $x$ does not contain any long runs, but a challenge for the \BMA algorithm is that the pointers in different traces will inevitably become misaligned if $x$  does contain a long run $0^\ell$ or $1^\ell$; this is why the \cite{BKKM04} algorithm must interleave \BMA with a procedure to handle long runs separately.
 Intuitively, deletion rate $\delta = n^{-1/2}$ is a barrier for the \cite{BKKM04} analysis because if $\delta = \omega(n^{-1/2})$, then each trace is likely to have multiple locations where more than one consecutive bit of $x$ is ``dropped,'' which is problematic for the analysis of \BMA given in \cite{BKKM04}.  
 
To summarize:  given this state of the art from prior work, it is clear that alignment-based approaches can outperform mean-based algorithms at low deletion rates, but it is not clear whether, or how far, alignment-based approaches can be extended beyond the \cite{BKKM04} results.  Further incentive for studying the low deletion rate regime comes from potential applications in areas such as computer networks, where it may be natural to model deletions as occurring at relatively low rates.  These considerations motivate the results of the present paper, which we now describe.

\subsection{This work:  An improved algorithm for the low deletion rate regime}

The main result of this paper is an efficient algorithm that can handle significantly higher deletion rates than the \cite{BKKM04} algorithm.\ignore{Before stating our main theorem, some useful notation:  we write ``$\bz \sim \Del_\delta(x)$'' to mean that $\bz$ is a trace generated by passing $x$ through the deletion channel of rate $\delta$.  }
We prove the following:

\begin{theorem} [Efficient trace reconstruction at deletion rate $\delta \geq n^{-(1/3 + \eps)}$]\label{thm:main}
  Fix any constant $\eps > 0$ and let $\delta = n^{-(1/3 + \eps)}.$ There is an algorithm {\tt Reconstruct} that uses $O(n^{4/3})$ 
  independent traces drawn from $\Del_\delta(x)$ \emph{(}where $x \in \zo^n$ is arbitrary and unknown to {\tt Reconstruct}\emph{)}, runs in $O(n^{7/3})$
  time, and outputs the unknown source string $x$ with probability at least $9/10$.
\end{theorem}

Note that any deletion rate $\delta < n^{-(1/3 + \eps)}$ can of course be handled, given \Cref{thm:main}, by simply deleting additional bits to reduce to the $\delta = n^{-(1/3 + \eps)}$ case.
Note further that any desired success probability $1-\kappa$ can easily be achieved from \Cref{thm:main} by running {\tt Reconstruct} $O(\log (1/\kappa))$ times and then taking a majority vote.   

At a high level, the {\tt Reconstruct} algorithm works by interleaving two different subroutines.
\begin{flushleft}
\begin{itemize}
\item The first subroutine is (essentially) the \BMA algorithm, for which we provide an improved analysis, showing that \BMA successfully reconstructs any string that does not contain a long subword
  (of length at least $M=2m+1$ with $m=n^{1/3}$) 
  that is a prefix of $s^{\infty}$ for some short (constant-length) bitstring $s$.
We refer to long and ``highly-repetitive''  subwords of $x$ of this form
  as ``\emph{$s$-deserts}'' of $x$; see \Cref{def:desert} for a detailed definition.

\item The second subroutine is a new algorithm which we show efficiently 
  determines the length of an $s$-desert in the source string $x$. 
\end{itemize}
\end{flushleft}
 Thus, two novel aspects of this work that go beyond \cite{BKKM04} are (i) our improved analysis of \BMA, and (ii) our new procedure for efficiently measuring deserts (the analogous component of the \cite{BKKM04} algorithm could only measure runs, which correspond to $s$-deserts with $|s|=1$).    
 
 We believe that it may be possible to further extend the kind of ``hybrid'' approach that we employ in this paper to obtain efficient trace reconstruction algorithms that can handle even larger deletion rates $\delta$.  However, there are some significant technical challenges that would need to be overcome in order to do so.  We describe some of these challenges at the end of the next section, which gives a more detailed overview of our approach.


\section{Overview of our approach} \label{sec:overview}
As alluded to in the introduction, at a high level our algorithm carries out a careful interleaving of two procedures, which we call \BMA and \DET. In this section we first give a high-level overview of the procedure \BMA as well as our improved analysis. Then we give a high-level overview of \DET, and finally we explain how these two procedures are interleaved. We close with a brief discussion of possibilities and barriers to further progress.

\subsection{Overview of \BMA}

The procedure \BMA is exactly the same as the bitwise majority alignment algorithm of \cite{BKKM04}; our new contribution regarding \BMA is in giving a more general analysis. To explain the high level idea, let us fix the deletion rate $\delta = n^{-(1/3 + \epsilon)}$ 
and a constant $C=$ $\lceil 100/\eps\rceil$.
Let 
\[
m=n^{1/3} \quad \text{and} \quad M=2m+1.
\]

 The \BMA procedure operates on a sample of some\ignore{\rnote{All the ``$\overallsc$'' 's (they are defined using a macro overallsc) in this section were ``$q$''.  In \Cref{sec:bma} we write ``$q$  but I'm not sure why; in the main algorithm in \Cref{sec:algorithm_overall} we write $\overallsc$. Let's hammer out what letter we want to use where - I guess we should be consistent throughout? I think I vote for ``$\overallsc$'' everywhere, which would mean changing ``$q$'' to the macro $\overallsc$'' in \Cref{sec:bma}. (This seems to go well with the closely related parameter $R$.)}}  $\overallsc=O(\log n)$ traces 
  $\by^1, \ldots, \by^\overallsc$ drawn independently from $\Del_\delta(x)$ before the procedure begins its execution. Note that for any $i \in [\overallsc]$ and 
any position $\ell$ in the trace $\by^i$, there is a position $f_i(\ell)$ satisfying $\ell\le f_i(\ell)\le n-1$ 
  in the target string $x=(x_0,\dots,x_{n-1})$ that maps to $\ell$ under the deletion process.\footnote{It will be convenient for us to index a binary string~$x \in \zo^\ell$
  using $[0:\ell-1]$ as $x=(x_0,\dots,x_{\ell-1})$.} 
The high level idea of \BMA is to maintain pointers $\current_1 ,\ldots , \current_\overallsc$,
  with $\current_i$ pointing to a position in $\by^i$, 
  such that \emph{most of them are correctly aligned} --- i.e., at the beginning of each time step $t$, 
  $t=0,1,\ldots$, 
  as we try to determine $x_t$,
we have $f_i(\current_i)= t$ for most $i \in [\overallsc]$.  Note that if this alignment guarantee were to hold for more than half of the traces for all $t=0,1,\ldots,n-1$, then we could reconstruct the unknown string $x$ by taking a majority vote of $\by^i_{\current_i}$ in each time step. 
 Indeed we show that this happens with high probability over 
  the randomness of $\by^1,\ldots,\by^\overallsc$ when $x$ does not contain an $s$-desert
  for any string $s\in \{0,1\}^{\le C}$ (i.e. a subword of length at least $M=2m+1$
  that is a prefix of $s^{\infty}$).
In contrast, the analysis of \cite{BKKM04} requires the deletion rate to be 
  $n^{-(1/2+\eps)}$ but works as long as $x$ does not contain a run of $0$'s or $1$'s
  (or $s$-deserts with $|s|=1$ in our notation)
  of length at least $\sqrt{n}$.

To explain \BMA in more detail, let us initialize $t=0$ and pointers 
  $\current_1(0),\ldots,\current_\overallsc(0)$ to position $0$.
(Note that most pointers are correctly aligned as desired 
  given that $\delta=n^{-(1/3+\eps)}$ and thus, $x_0$ is not deleted in most traces and  $f_i(\current_i(0))=f_i(0)=0$ for most
  $i$.)
The way the pointers are updated is as follows: At each time step $t$, we let $w_t$ be the majority element of the $\overallsc$-element multiset $\smash{\{\by^i_{\current_i(t)}\}_{i\in [\overallsc]}}$.  
For those traces $\by^i$ with $\smash{\by^i_{\current_i(t)} = w_t}$ (i.e., the bit of $\by^i$ at the current pointer is the majority bit), we move the pointer to the right by $1$, i.e.~we set $\current_i(t+1)\leftarrow \current_i(t)+1$; otherwise the
  pointer stays the same, i.e., we set $\current_i(t+1) \leftarrow \current_i(t)$.
Next we increment $t$ and start the next round, repeating until $t=n$ when
  $\BMA$ outputs the string $(w_0,\ldots,w_{n-1})$. 
 
 For intuition we observe that if most of the pointers were aligned at the beginning of time step $t$ (i.e.,
   $f_i(\current_i(t))=t$ for most $i\in [\overallsc]$), then $w_t=x_t$ is indeed the next bit in $x$. 
Moreover, if $\current_i(t)$ is aligned and $w_t=x_t$, then 
  moving $\current_i$ to the right by $1$ is justified by noting that
  most likely $x_{t+1}$ is not deleted in $\by^i$ (with probability $1-\delta$),
  and when this happens
  $f_i(\current_i(t+1))=t+1$ so  $\current_i$ remains aligned at the beginning of the next time step. 
  
In more detail, our analysis shows that when  $x$ does not contain
  an $s$-desert for any $s\in \{0,1\}^{\le C}$ 
  $\BMA$ maintains the following invariants at the beginning of time step $t=0,1,\ldots,
  n $:
\begin{enumerate}
\item At time $t$, $\BMA$ has reconstructed $x_0, \ldots, x_{t-1}$ 
  correctly as $w_0,\ldots,w_{t-1}$. 
\item For every trace $\by^i$, $i\in [\overallsc]$, it holds that $ f_i(\current_i(t))\ge t$. 
\item Finally, $\sum_{i\in [\overallsc]} \big(f_i(\current_i(t))-t\big)\le  2\overallsc/C$.
\end{enumerate}
The intuitive meaning behind conditions (2) and (3) is as follows: while (2) says that 
  the ``original position'' of $\current_i(t)$ never falls behind $t$, condition (3) ensures that on average, the original positions of these pointers do not surpass $t$ by too much. In fact, since $C$ is a large constant, most of the pointers are perfectly aligned, i.e.,~they satisfy $f_i(\current_i(t))=t.$


We now discuss how the invariants (1), (2) and (3) are maintained. First, we observe that invariant (1) 
  for time step $t+1$, i.e., $w_t=x_t$, follows immediately from  (3) at time step $t$.
Invariant (2) for time step $t+1$ follows almost immediately from (2) at $t$ and $w_t=x_t$.
(If $f_i(\current_i(t))>t$, then $f_i(\current_i(t+1))\ge f_i(\current_i(t))\ge t+1$ given that both
  $f_i$ and $\current_i$ are nondecreasing; if $f_i(\current_i(t))=t$ is aligned at time step $t$, then
  $w_t=x_t$ implies $\current_i(t+1)=\current_i(t)+1$ and thus, $f_i(\current_i(t+1))\ge t+1$.)
The main challenge is to show that invariant (3) is maintained.
While this is not true for~a~general string $x$, we show that this holds  with high probability (over $\by^1,\ldots,\by^\overallsc\sim \Del_\delta(x)$) for any string $x$ which does not have an $s$-desert for any $s\in \{0,1\}^{\le C}$.
  (We note here that the value of $m$ is selected so as to satisfy $m\delta\ll 1$; on the other hand, when we discuss the \DET procedure below, we will see that we also require $m$ to satisfy $m \gg \sqrt{\delta n}$.)

In a nutshell, the main proof idea for (3) is to exploit the fact that when we draw     
  $\by^1,\ldots,\by^\overallsc \sim \Del_{\delta}(x)$, with high probability they satisfy two properties:
(i) for every $\by^i$ and
  every subword of roughly $C^2m$ consecutive positions in the original string $x$, no more than $C$ positions
 within the subword are deleted in the generation of $\by^i$;
(ii) for every subword of roughly $m$ consecutive positions in $x$,
  the number of $\by^i$ that have at least one deletion in the subword is no more than $\overallsc/C^3$.
These two properties can be shown using straightforward probabilistic arguments by taking advantage of the aforementioned 
  $m\delta\ll 1$.\ignore{ (e.g., this implies that for any window of length $m$,
  the probability of a trace having at least one deletion is at most $m\delta=1/n^{\eps}$; 
  it follows from a Chernoff bound that the second property holds with high probability).}
Using these two properties, a detailed (non-probabilistic) argument shows that  \BMA can reconstruct the string $x$ with high probability if $x$ contains no $s$-desert. 

The above discussion sketches our argument  
that if the target string $x$ does not have an $s$-desert, then \BMA correctly reconstructs  $x$. More generally, our arguments show that if $x$ does have an $s$-desert, then \BMA correctly
  reconstructs the prefix of $x$ up to the position when an $s$-desert shows up:
Let $r$ be  
  the first position in $x$ that is 
  ``\emph{deep in an $s$-desert}''; this is the first position in $x$ such that
  $x_{[r-m:r+m]}$\footnote{For a string $x \in [0:n-1]$ integers $0 \leq a < b \leq n-1$, we write $x_{[a:b]}$ to denote the \emph{subword} $(x_a, x_{a+1}, \dots, x_b)$.},   
  the length-$M$ subword of $x$ 
  centered at $r$, is an $s$-desert. Then
  \BMA correctly reconstructs the prefix of $x$ up to position $r+m$.
Having reached such a position, it is natural to now ask --- ``how do we determine the \emph{end} of this desert?''. This naturally leads us to the next procedure \DET.

\ignore{
}


\subsection{Overview of \DET} 
Suppose that $x$ has an $s$-desert with $|s|=k\le C$, so \BMA reconstructs the length-$(r+m+1)$ prefix of $x$,
  where $r$ is the first position that is ``deep in the $s$-desert'' (note that it is easy to determine the position $r$ from the output of \BMA).  The algorithm \DET takes as input the prefix $x_{[0:r+m]}$ of $x$ and the location $r$, 
  and its task is to compute the end of the $s$-desert:
  the first position $\dend\ge r+m$ such that 
  $x_{\dend+1}\ne x_{\dend-k+1}$.
The \DET algorithm is rather involved but at a high level it consists of two
  stages: an initial \emph{coarse estimation} 
  of the end of the desert followed by 
  \emph{alignments} of traces from $\Del_\delta(x)$ 
  with the end of the desert 
  (using the coarse estimate).

\medskip

{\bf Coarse estimation:} The goal of the coarse estimation stage is to 
  identify an integer $\hat{\beta}$ that~is close to $(1-\delta)\dend$:
  $|\hat{\beta}-(1-\delta)\dend|\le 2\sigma$, where $\sigma := \tilde{O}(\sqrt{\delta n})\ll m$
  is basically how far an entry $x_i$ of $x$ can deviate from its expected location $(1-\delta)i$ in a typical trace
  $\by\sim \Del_\delta(x)$.
Intuitively, $\hat{\beta}$ is an estimation of the location of $x_{\dend}$
in a trace $\by\sim \Del_\delta(x)$ that contains it, i.e., when $x_{\dend}$ is not deleted in $\by$.
To do this, we draw $\alpha=O(1/\eps)$ many traces $\by^1,\dots,\by^\alpha\sim \Del_\delta(x)$.
Roughly speaking, we split each~trace $\by^i$ into overlapping intervals of length $4\sigma$.
The first interval starts at $(1-\delta)r$ and each successive 
interval shifts to the right by $\sigma$ (so it 
  overlaps with the previous interval by $3\sigma$).
Since 
	$m\gg \sigma = \tilde{O}(\sqrt{\delta n})$ (which is
one of the bottlenecks that requires $\delta \ll n^{-1/3}$), 
  the $s$-desert is unlikely to end before $(1-\delta)r$ in a trace $\by\sim \Del_\delta(x)$
  and must end in one of constantly many 
   intervals with high probability, by the choice of $\sigma$.
To identify one such interval, we make the following observation.
Let $\Cyc$ be the set of all $k$-bit strings that can be obtained as cyclic shifts of $s$.
Given $\dend$ as the end of the $s$-desert that starts
  at $x_{r-m}$, every $k$-bit subword of $x_{[r-m:\dend]}$ is in $\Cyc$
  but $(x_{\dend-k+2},\ldots,x_{\dend+1})\notin \Cyc$, and these $k\le C$ bits
  will most likely remain in a trace given the low deletion rate.
This motivates us to look for the leftmost interval $I^*$ such that in at least half of $\by^1,\dots,\by^\alpha$, 
it holds that $\by^i_{I^*}$ contains a $k$-bit subword not in $\Cyc$.
We show that 
  with high probability, setting $\smash{\hat{\beta}}$ to be
  the right endpoint of $I^*$ gives us 
    a coarse estimate of $(1-\delta)\dend$ up to an accuracy of $\pm 2\sigma$.

In addition to obtaining $\hat{\beta}$,
  the coarse estimation stage recovers the following $8\sigma$-bit subword of $x$:
  $(x_{\dend-k+2},\ldots,x_{\dend-k+8\sigma+1})$, which we will refer to as the 
  \emph{tail string}
  of the $s$-desert and denote by $\tail\in \{0,1\}^{8\sigma}$.
To this end, we draw another $\alpha=O(1/\eps)$ fresh traces $\by^1,\ldots,\by^\alpha$ and examine
  the subword of each $\by^i$ of length $6\sigma$ centered at location $\hat{\beta}$.
Each $\by^i$ looks for the first $k$-bit subword in this interval that is not in $\Cyc$
  and votes for its $8\sigma$-bit subword that starts at this non-cyclic shift as its candidate for $\tail$.
We show that with high probability, the string with the highest votes is exactly $\tail$.
(We note that both parts of this coarse estimation procedure require that with high probability, any fixed interval of length $O(\sigma)$ in $x$ does not get any deletions in a random trace, i.e., $\sigma \delta\ll 1$. This follows from the two
  constraints $m\delta\ll 1$ and $m\gg \sigma$.)

\medskip

{\bf Alignments:} 
Suppose the first stage succeeds in computing $\hat{\beta}$ and $\tail \in \{0,1\}^{8\sigma}$.
The second stage is based on a procedure called $\Align$ 
which satisfies two crucial criteria. These criteria are as follows:  if $\Align$ is given an input trace $\by \sim \Del_\delta(x)$, then (a) with \emph{fairly high} probability (by which we mean 
$1-n^{-\Theta(\eps)}$ for 
the rest of the overview) it returns a location $\ell$ in $\by$ such that 
 $\by_{\ell}$ corresponds to 
 $x_{\dend}$ in $x$, 
and moreover (b) the expectation of $\ell$ (over the randomness of $\by$)   is a ``sharp estimate'' of $(1-\delta)\dend$
  that is accurate up to an additive $\pm 0.1$ error.\footnote{We note that item (b) is not an immediate consequence of item (a). In more detail, the failure probability of (a) is roughly $1/n^{\Theta(\eps)}$, but if when $\Align$ fails in (a) it returns a location that is inaccurate by $\gg n^{\Theta(\eps)}$ positions, then (b) would not follow from (a). 
Indeed significantly more effort is required in our analysis to ensure (b).} 
To pin down the exact end of the $s$-desert, 
  $\DET$ simply draws $\smash{\tilde{O}(n^{2/3-\eps})}$ many
  traces, runs $\Align$ on each of them and computes the average $\smash{\hat{\ell}}$ of the locations
  it returns.
It is easy to show that rounding $\hat{\ell}/(1-\delta)$ to the nearest integer gives $\dend$ with high probability.

{
The case when $k=|s|=1$ (so the desert is a long subword consisting either of all 0's or of all 1's) is significantly easier (and was implicitly handled in \cite{BKKM04}), so in the following discussion we focus on the case when $k=|s|\ge 2$ and the desert has a more challenging structure.  For this case our $\Align$ procedure uses a new idea, which is that of a ``\emph{signature}.'' 
A \emph{signature} is a subword of $x$, denoted $\sig$, of length between $2k$ and $8\sigma$ that starts at the 
  same location $x_{\dend-k+2}$ as $\tail$ (so $\sig$ is contained in $\tail$, since $|\tail|=8\sigma$) and either ends at a location  $d$ which is the smallest integer $d\in [\dend+k+1:\dend+8\sigma-k+1]$
  such that the $k$-bit subword that ends at $d$ is not in $\Cyc$, or has length $8 \sigma$
  if no such $d$ exists (in this case $\sig$ is the same as $\tail$). We remind the reader that 
  the first $k$-bits of $\tail$, and hence also of $\sig$, is a string not in $\Cyc$, and the same is true of the last $k$ bits of $\sig$ if its length is less than $8 \sigma.$

Given a trace $\by\sim \Del_\delta(x)$, \Align (roughly speaking) attempts
  to locate the image of $x_\dend$ in $\by$ by locating the image of $\sig$ within an interval in 
  $\by$ of length $O(\sigma)$ around $\smash{\hat{\beta}}$.
  In a bit more detail, it checks whether the restriction of $\by$ to a certain interval $J$ around $\smash{\hat{\beta}}$ is of the form $w \circ \sig \circ v$, such that the first $k$ bits of $\sig$ is the leftmost $k$-bit subword of $\by_J$ that is not in $\Cyc$.
  If $\by$ does not satisfy this condition then \Align discards that trace and outputs \nil.
  We note that if the only goal of \Align were to locate a position $\ell$ in $\by$ such that with fairly high probability $\by_\ell$ corresponds to $x_\dend$ (i.e.~item (a) above), then in all other cases (i.e.~whenever $\by$ does satisfy the above condition) \Align could return the index of the $(k-1)$-th bit of $\sig$ in $\by_J$.
(By doing this, $\Align$ always returns the correct position whenever 
  the subword of $x$ of length $O(\sigma)$ centered at $\dend$ has no deletion in $\by$
  and $x_\dend$ deviates from its expected location in a trace by at most 
  $\sigma$ in $\by$, which happens with probability $O(\sigma\delta)=n^{-\Theta(\eps)}$.)  
However, it turns out that $\Align$ must proceed in a slightly (but crucially) different way in order to additionally satisfy item (b) above (i.e.,~have the expected value of its output locations provide an accurate ``sharp estimate'' of $(1-\delta)\dend$).
The actual execution of \Align is that in the case when $\by_J$ does satisfy the above condition, 
  \Align returns the index of the $(k-1)$-th bit of $\sig$ in $\by_J$ with high probability and with some small remaining probability (the precise value of which depends on the location of $\sig$ within $\by_J$), \Align opts to still output $\nil.$ 
  A detailed analysis, which we provide in \Cref{sec:kistwo}, shows that this \Align procedure satisfies both criteria (a) and (b) described above.
}

\subsection{The overall algorithm}
The overall algorithm works by alternately running \BMA and \DET. It starts with \BMA, which draws $N=O(\log n)$ traces of $x$ and returns
  the first position $r$ in $x$ that is deep in a desert as well as
 the prefix $w=x_{[0:r+m]}$ of the target string $x$.
Then the algorithm runs
  $\DET$ to compute \dend, the right end of the desert.
Note that the execution of \BMA will misalign some small fraction of the traces it uses, but these errors do not affect \DET as \DET is run using fresh traces.

With \dend\ from \DET, the algorithm has now reconstructed the prefix 
  $x_{[0:\dend]}$ by extending $x_{[0:r+m]}$.
Next the algorithm runs $\BMA$ again on $N$ traces that are, ideally, drawn from 
  $x_{[\dend+1:n-1]}$, in order to reconstruct the next segment of $x$ until a new desert shows up (at which point 
  the algorithm repeats).
These traces are obtained by running the $\Align$ procedure used by $\DET$ on
  $N$ fresh traces $\by^1,\ldots,\by^N$ of $x$.
Let $\ell_i$ be the output of $\Align$ running on $\by^i$.
As noted in (a) earlier,  
  all but a small fraction of $\ell_i$'s
  are such that the desert ends at $\smash{\by^i_{\ell_i}}$ in $\by^i$.
We then run $\BMA$ on $\smash{\bz^1,\ldots,\bz^N}$, where
  $\bz^i$ is the suffix of $\by^i$ starting at $\ell_i+1$ for each $\smash{i}$.
Even though $\smash{\bz^1,\ldots,\bz^N}$ are \emph{not} exactly $N$ fresh traces of $x_{[\dend+1:n-1]}$
  (since a small and arbitrary fraction of $\by^i$ might be misaligned), 
   \BMA is able to succeed because of a crucial \emph{robustness} property.   This property is that the correctness guarantee of \BMA holds even when a small and ``adversarially'' picked constant fraction of the $N$ traces given to it are misaligned; intuitively, \BMA enjoys this robustness because it works in each time step by taking a majority vote over its input traces, so as long as a substantial majority of the traces are correctly aligned, even a small constant fraction of adversarial traces cannot affect its correctness. The algorithm continues alternating between \BMA and \DET, and is thereby able to reconstruct the entire target string $x$.

\subsection{Discussion}

We believe that it may be possible to improve on our algorithmic results (specifically, to handle larger $1/\poly(n)$ deletion rates) by refining our basic approach of alternating between successively (i) exactly inferring non-desert subwords of the string (as currently done by \BMA), and (ii) exactly computing the length of the intervening deserts 
  (as currently done by \DET).
However, some nontrivial obstacles are encountered in attempting to go beyond $\delta = n^{-1/3}$ to improve the current result, which we now briefly discuss.

Roughly speaking, our algorithm takes advantage of the  fact that with $\delta = n^{-(1/3 + \eps)}$, we have $\sqrt{\delta n}  \ll 1/\delta$ and as a result it is easy to come up with the initial ``coarse estimate'' of the right end of a particular desert given fresh traces as the input to \DET.   If $\delta = n^{-(1/3 - \eps)}$ then we have $\sqrt{\delta n} \gg 1/\delta$ and it is not  clear how a procedure like \DET can reliably find any location within a particular desert of interest given fresh traces (which seems to be a prerequisite to performing ``coarse estimation'').
One possible way around this is to not draw fresh traces at the start of each execution of \DET, but if the same traces are used for repeated runs of both \BMA and
\DET then it is not clear how to prevent alignment errors from accumulating, and such accumulated errors seem problematic for running \DET successfully. A natural approach to pursue along these lines is to  develop an improved version of \DET which can handle some substantial number of ``misaligned'' or ``adversarial'' input traces; this is an interesting direction for future work.




\section{Preliminaries} \label{sec:preliminaries}

\noindent {\bf Notation.}
Given a positive integer $n$, we write $[n]$ to denote $\{1,\ldots,n\}$.
Given two integers $a\le b$ we write $[a:b]$ to denote $\{a,\ldots,b\}$. 
We write $\ln$ to denote natural logarithm and $\log$ to denote logarithm to the base 2.
We denote the set of non-negative integers by $\Z_{\geq 0}$.
We  write ``$a=b\pm c$'' to indicate that $b-c\le a\le b+c$.


\medskip

\noindent {\bf Subword.} 
It will be convenient for us to index a binary string~$x \in \zo^n$
  using $[0:n-1]$ as~$x=$ $(x_0,\dots,x_{n-1})$. 
Given such a string $x\in \{0,1\}^n$ and integers $0 \leq a\le  b \leq n-1$, we write $x_{[a:b]}$ to denote the \emph{subword} $(x_a, x_{a+1}, \dots, x_b)$ of $x$.
An  \emph{$\ell$-subword}   of $x$ is a 
  subword of $x$ of length $\ell$, given by $(x_a, x_{a+1}, \dots, x_{a+\ell-1})$ 
  for some $a \in [0: n-\ell]$. 
 
\medskip

\noindent {\bf Distributions.}
When we use bold font such as $\bD, \by, \bz$, etc., it is to emphasize that the entity in question is a random variable.
We write ``$\bx \sim {\cal D}$'' to indicate that random variable~$\bx$~is 
  distributed according to distribution ${\cal D}$.

\medskip

\noindent {\bf Deletion channel and traces.}
Throughout this paper the parameter $\delta:0 <$ $\delta < 1$ denotes~the \emph{deletion probability}.  Given a string $x \in \zo^n$, we write $\Del_\delta(x)$ to denote the distribution of the string that results from passing  $x$ through the $\delta$-deletion channel (so the distribution $\Del_\delta(x)$ is supported on $\zo^{\leq n}$), and we refer to a string in the support of $\Del_\delta(x)$ as a \emph{trace} of $x$.  Recall that a random trace $\by \sim \Del_\delta(x)$ is obtained by independently deleting each bit of $x$ with probability $\delta$ and concatenating the surviving bits.\hspace{0.05cm}\footnote{For simplicity in this work we assume that the deletion probability $\delta$ is known to the reconstruction algorithm.  We~note that it is possible to obtain a high-accuracy estimate of $\delta$ simply by measuring the average length of traces received from the deletion channel.}

\medskip

\noindent {\bf A notational convention.} In several places we use $\mathsf{sans \ serif \ font}$ for names such as $\tail$ (which is a subword of the target string $x$), $\dend$ (which is a location in the target string $x$), and so on.  To aid the reader, whenever we use this font the corresponding entity  is an ``$x$-entity,'' i.e.~a location, subword, etc.~that is associated with the source string $x$ rather than with a trace of $x$.


\section{The main algorithm} \label{sec:algorithm_overall}
In this section we describe the main algorithm {\tt Reconstruct}.
We begin by giving a precise definition of the notion of an \emph{$s$-desert}. To do this, here and throughout the paper we fix 
\[
C := \lceil 100/\epsilon \rceil,
\quad \text{and we recall that} \quad
m = n^{1/3}\quad\text{and}\quad M = 2m+1.
\]

\begin{definition} \label{def:desert}
For $s\in \{0,1\}^{\le \deslen}$, a binary string $z\in \{0,1\}^*$ is said to be an \emph{$s$-desert}
  if $z$ is a prefix of $s^{\infty}$ and $|z|\ge M$. 
A string is said to be a \emph{desert} if it is an $s$-desert for some $s\in \{0,1\}^{\le \deslen}$.
Given a string $x\in \{0,1\}^n$, we say that a location $i\in [0:n-1]$ is 
\emph{deep in a desert} if the length-$M$ subword $x_{[i-m:i+m]}$ centered at $i$ is a desert.
We say a string $x$ has \emph{no desert} if
  no subword of $x$ is a desert (or equivalently,  no location $i\in [0:n-1]$ is deep in a desert in $x$);
  otherwise we say that it has at least one desert.\ignore{\footnote{Xi: I was hoping for 
  something like, given any $x$, there are disjoint intervals $[a_1:b_1],[a_2:b_2],\ldots$
  such that $x_{[a_i:b_i]}$ are the maximal  desert subwords in $x$. 
However, I think two maximal desert subwords of $x$ could overlap by a little bit, 
  say length at most $2000$. 
On the other hand, given any $x$, the beginning of its first desert is well defined
  and this will be what we use below.}}
\end{definition}

\subsection{The preprocessing step}  \label{sec:preprocess}

Before stating the main algorithm, we first describe a preprocessing step. This is a simple routine which we call {\tt Preprocess} and which is described and analyzed in~\Cref{lem:preprocess} in \Cref{sec:ap-preprocess}. Its goal is to output a string $v\in \{0,1\}^{n/2}$ 
and the main feature of $v$ is that it lets us assume that any $s$-desert (for any $s \in \{0,1\}^{\le \deslen})$ of $z:=x\circ v$ ends at least $n/2-(2m+1)$ bits before the right end of the string $z$.
{\tt Preprocess} succeeds in finding such a string $v$ with probability at least $1-n^{-\omega(1)}$ so we 
  assume below that $v$ satisfies this property.
Further, given $v$, we can simulate sample access to $\Del_\delta(z)$
  using that to $\Del_\delta(x)$ trace by trace (with a linear overhead in running time).
In order to obtain the original $n$-bit string $x$ it suffices for us to reconstruct the $(3n/2)$-bit string $z$.

For convenience of notation, we rename $z$ as $x$ and rename $n$ to be the length of this string $z$, so we still have $x=(x_0,\dots,x_{n-1})$. Now $x$ is an $n$-bit string that has the following property: any desert in $x$ ends at least $n/4$ bits before the right end of $x$. With this preprocessing accomplished, we now describe Algorithm {\tt Reconstruct} in \Cref{fig:mainalg}.

\begin{figure}[!t]
\setstretch{1.2}
  \begin{algorithm}[H]
    \caption{Algorithm {\tt Reconstruct} for $\delta = n^{-(1/3+\eps)}$}\label{fig:mainalg}
		\DontPrintSemicolon
		\SetNoFillComment
		\KwIn{Length $n$ of an unknown $x \in \{0,1\}^n$ and
   access to $\Del_\delta(x)$ where $\delta = n^{-(1/3+\eps)}$}
		\KwOut{A string $u$, where the algorithm succeeds if $u=x$}
Set $N := O(\log n)$ \\	
		Draw $N$ fresh traces\ignore{\rnote{I changed the name - we had been calling these $\by^1,\dots$ which led to some naming confusion. Now the traces we give to \BMA are always called $\bz^i$}} $\bz^1,\ldots,\bz^N$ independently from $\Del_\delta(x)$\\
		Run $\BMA(n, \{\bz^1,\ldots,\bz^N\})$ and let $w$ be its output\\
    \lIf{$w$ has no desert}{\Return $w$}\
    \Else{
		  Let $r$ be the first location that is deep in a desert in $w$ and let $u=w_{[0:r+m]}$
    }
    \tcp{Main loop}
		\For{$n/m$ rounds}{
		Draw $N$ fresh traces $\by^1,\ldots,\by^N$ independently from $\Del_\delta(x)$\\
		Run $\DET(r,u,\{\by^1,\ldots,\by^N\})$ and let $b$ and $\ell_i$, $i\in [N]$, 
		  be its output\\
		Set $r=b$ and extend $u$ to be a string of length $b$ such that $u_{[r-m:b]}$ is a desert\\
		    \lIf{$b=n-1$}{output ``\emph{FAILURE}''}
		{Let $\bz^i$ be the suffix of $\by^i$ starting at $\by^i_{\ell_i+1}$ for each $i\in [N]$\\
		 Run $\BMA(n-b-1,\{\bz^1,\ldots,\bz^N\})$ and let $w$ be its output\\ 
		 \If{$w$ has no desert}{
		   \Return $u\circ w$
     }\Else{Let $r^*$ be the first location that is deep in a desert in $w$ and set
		  $r\leftarrow r+r^*$ and $u\leftarrow u\circ w_{[r^*+m]}$\\
		  \Return $u$ if $u$ is of length $n$}
		}
		}
		\Return $u$ 
  \end{algorithm}
\end{figure} 

\subsection{The high level idea of the {\tt Reconstruct} algorithm}

At a high level the algorithm works as follows.
It starts (lines 1-3) by drawing 
\[
N=O(\log n)
\]
independent traces $\bz^1,\ldots,\bz^N$  from $\Del_\delta(x)$ and using them to run \BMA. 
An important component of our analysis is the following new result about the performance of \BMA (note that later we require, and will give, a more robust version of the theorem below; see \Cref{thm:realBMA}):

\begin{theorem} \label{thm:BMA}
Let $\delta=n^{-(1/3+\eps)}$ for some fixed constant $\eps>0$.
Given $N$ traces drawn independently from $\Del_\delta(x)$ for some unknown string $x\in \{0,1\}^n$, \BMA
  runs in $\tilde{O}(n)$ time and returns a string $w$ of length $n$ 
 with the following performance guarantees: 
   \begin{flushleft}\begin{enumerate}
     \item If $x$ has no desert, then $w=x$ with probability at least $1-1/n^2$;
     \item If $x$ has at least one desert, then $w$ and $x$ share the same $(r+m+1)$-bit prefix with probability at least $1-1/n^2$, where $r$ is the first location that is deep in a desert 
       of $x$.
   \end{enumerate}\end{flushleft}
\end{theorem}

Let $w$ be the string \BMA returns.
By \Cref{thm:BMA}, we have the following two cases:
\begin{flushleft}\begin{enumerate}
\item If $w$ has no desert, then also $x$ has no desert and the algorithm can just return $w$ (line 4);
\item If $w$ has at least one desert, then writing $r$ to denote the first
  location that is deep in a desert in $w$, it is safe to assume that
  $w_{[0:r+m]}=x_{[0:r+m]}$ and 
  $r$ is also the first location\\ that is deep in a desert in $x$ (line 6).
\end{enumerate}\end{flushleft}

Suppose that we are in the second case with $\smash{w_{[0:r+m]}=x_{[0:r+m]}}$.
Then $w_{[r-m:r+m]}$ is an $s$-desert for some string $s\in \{0,1\}^k$ of some length $k\le \deslen$.
We let $s$ be the shortest such string and let $k$ be its length (so if $w_{[r-m:r+m]}$ were, for example, a subword of the form 001001001001$\cdots$ of length a multiple of 12, we would take $s=001$ and $k=3$).

Next (lines 8-9) we run $\DET$ to figure out where this repetition of $s$ ends in $x$.
We use $\dend$ to denote the end of the desert, where $\dend\ge r+m$ is the smallest integer such that $x_{\dend+1}\ne x_{\dend-k+1}$.
By the preprocessing step we may assume that $\dend$ exists and satisfies $\dend\le 3n/4$.
(We note that $\DET$ has access to $\Del_\delta(x)$ to draw fresh traces by itself; 
  we send $N$ fresh traces $\by^1,\ldots,\by^N$ to $\DET$ so that it can help align them
 to the end of the desert, which are used to run \BMA later.)
 The performance guarantee for \DET is given below:

\begin{theorem} \label{thm:DET}
  Let $\delta = n^{-(1/3+\eps)}$ for some fixed constant $\eps>0$. There is an algorithm $\DET$ with the following input and output:
\begin{flushleft}
\begin{itemize}
\item {\bf Input:} (i) a location $r \in [0:3n/4]$,
  (ii) a string $u\in \{0,1\}^{r+m+1}$, 
  (iii) a multiset of
  strings $\{y^1,\ldots,y^{\overallsc}\}$ from $\{0,1\}^{\le n}$ where $\overallsc=O(\log n)$, and (iv) sample access
  to $\Del_\delta(x)$ for some unknown string  $x\in \{0,1\}^n$.

\item {\bf Output:} An integer $b$, and an integer $\ell_i$ for each $i\in [\overallsc]$.
\end{itemize}
\end{flushleft}
The algorithm \DET draws $\tilde{O}(n^{2/3-\eps})$ many independent 
traces from $\Del_\delta(x)$, runs in $O(n^{5/3})$ time and
  has the following performance guarantee.
Suppose that $r$ is the first location that is deep in~some desert of $x$; $u= x_{[0:r+m]}$; the unknown $\dend$ of the desert to which $x_r$ belongs
  is at most  $3n/4$; and $\smash{y^1=\by^1,\ldots,y^\overallsc=\by^{\overallsc}}$ 
  are independent traces drawn from $\Del_\delta(x)$. 
  Then the integers $b$ and $\ell_i$ that $\DET$ outputs 
  satisfy the following properties with probability
  at least $1-1/n^2$:
  $b=\dend$, and $\ell_i=\pos(\by^i)$ for at least $0.9$ fraction of $i\in [\overallsc]$.
  Here $\pos(y)$ for a trace $y$ denotes the location $\ell$
  in $y$ such that $\smash{y_\ell}$ corresponds to the last bit of $\smash{x_{[0:\dend]}}$ that survives in $y$ \emph{(}and we set $\smash{\pos(y)=-1}$ by default if all of $x_{[0:\dend]}$ gets deleted in $y$\emph{)}.
\end{theorem}

Line~9 runs $\DET$ with fresh independent traces
  $\by^1,\ldots,\by^N$ drawn from $\Del_\delta(x)$.
By \Cref{thm:DET}, with high probability \DET returns the correct location $b=\dend$, from which we can then recover $x_{[0:b]}$ as the unique extension of $w_{[0:r+m]}$ in which the pattern $s$ keeps repeating until (and including) location $b$. 
Moreover, we have from \Cref{thm:DET} that, for at least a $9/10$-fraction of all $i\in [N]$, the suffix $\smash{\bz^i}$ of $\smash{\by^i}$ starting from $\smash{\by^i_{\ell_i+1}}$ is a trace drawn from $\Del_\delta(x_{[b+1:n-1]})$.   
We further note that  our preprocessing ensures $\smash{b \le 3n/4}$ and thus,
   the algorithm does not halt on line 11.


To continue, we would like to run \BMA again on $\bz^1,\ldots,\bz^N$ (the suffixes of $\by^1,\dots,\by^N$) to recover $x_{[b+1:n-1]}$ (or a prefix of $x_{[b+1:n-1]}$ if it contains a desert).  However, observe that now we need \BMA to be robust against some noise in its input traces because by \Cref{thm:DET}, up to $1/10$  of $\smash{\bz^1,\ldots,\bz^N}$ might have been obtained from an incorrect alignment of $\by^1,\dots,\by^N.$
Thus we require the following more robust performance guarantee from \BMA, given by \Cref{thm:realBMA} below.  (To state this we need a quick definition: we say two multisets of strings of the same size are \emph{$\eta$-close} if one can be obtained
  from the other by substituting no more than $\eta$-fraction of its strings.
One should also consider $x'$ in the statement below as $x_{[b+1:n-1]}$ and $n'$ as $n-b-1$.)

\begin{theorem} \label{thm:realBMA}
Let $\delta=n^{-(1/3+\eps)}$ for some fixed constant $\eps>0$.
Suppose $\tilde{\bz}^{1},\ldots,\tilde{\bz}^{N}$ are $N$ independent traces drawn from $\Del_{\delta}(x')$ for some unknown string $\smash{x'\in \{0,1\}^{n'}}$ with $n' \leq n$.
The following holds with probability at least $1-1/n^2$ over the randomness of $\smash{\tilde{\bz}^1,\ldots,\tilde{\bz}^\overallsc}\sim \Del_\delta(x')$: 
  \begin{flushleft}\begin{enumerate}
    \item If $x'$ has no desert, then  $\BMA$ running on 
      $n'$ and any multiset $\{z^1,\ldots,z^\overallsc\}$ that is $(1/10)$-close to %
  $\{\tilde{\bz}^1,\ldots,\tilde{\bz}^\overallsc\}$ returns $w=x'$; 
    \item If $x'$ has at least one desert, then 
    $\BMA$ running on 
      $n'$ and any multiset $\{z^1,\ldots,z^\overallsc\}$ that \\is $(1/10)$-close to %
  $\{\tilde{\bz}^1,\ldots,\tilde{\bz}^\overallsc\}$ returns a string $w$ that shares the 
    same $(r'+m+1)$-bit\\ prefix with $x'$, where $r'$ is the first location that is deep in a 
    desert in $x'$. 
  \end{enumerate}\end{flushleft}
\ignore{
  The algorithm \BMA has the following performance guarantees.
  Let $x\in \{0,1\}^n$ and $\bz^{1\ast},\ldots,\bz^{N\ast}$ be $N$ traces drawn independently from $\Del_\delta(x)$.
  With probability at least $1-1/n^C$ (over draws of $\bz^{1\ast},\ldots,\bz^{N\ast}$ from $\Del_\delta(x)$),
  \begin{flushleft}\begin{enumerate}
    \item if $x$ has no desert, then \BMA running on any multiset that is $(1/n^{\kappa})$-close to $\{ \bz^1,\ldots,\bz^N\}$ returns $x$;
    \item if $x$ has at least one desert (letting $r\in [n]$ be the first location that is deep in a desert 
      of $x$), then \BMA running on any multiset that is $(1/n^{\kappa})$-close to $\{ \bz^1,\ldots,\bz^N\}$ returns a string that shares the same $(r+m)$-bit prefix with $x$.
  \end{enumerate}\end{flushleft}}
\end{theorem}

Given \Cref{thm:realBMA}, we can indeed successfully run \BMA on $\bz^1,\ldots,\bz^N$ and with high probability, it correctly recovers a prefix of $x_{[b+1:n-1]}$ up to the first point deep in the next desert (if any exists), in which case the algorithm repeats (if there is no next desert, then with high probability \BMA will correctly recover the rest of $x$).

\subsection{Correctness of {\tt Reconstruct}}

The case when $x$ has no desert is handled by \Cref{thm:realBMA}.
Assuming that $x$ has at least one desert, it follows from \Cref{thm:realBMA}
  that $r,u$ together  satisfy the following property with probability at
  least $1-1/n^{2}$ at the beginning of the main loop (lines~7-18):
$r$ is the first location that is deep in a desert in $x$ and $u=x_{[0:r+m]}$.
This gives the base case for the following invariant that the algorithm maintains 
  with high probability:
\begin{quote}
{\bf Invariant:} \emph{At the beginning of each loop, $r$ is the first location deep in some desert in $x$ and $u=x_{[0:r+m]}$.}
\end{quote}

Assume that the invariant is met at the beginning of the current loop. Let $\dend$ denote the end of the current desert (i.e., the smallest value $\dend \ge r+m$ such that $x_{\dend+1}\ne x_{\dend-k+1}$; we observe that $\dend\le 3n/4$ {always exists by the guarantee 
  of the preprocessing step}).
Let $\by^1,\ldots,\by^N$ be fresh traces drawn at the beginning of this loop.
For each $i\in [N]$, 
we write $\tilde{\bz}^i$ to denote the suffix~of~$\by^i$ starting at $ \pos(\by^i)+1$.
Given that $\by^1,\ldots,\by^N\sim\Del_\delta(x)$,
  $\tilde{\bz}^1,\ldots,\tilde{\bz}^N$ are indeed independent traces drawn from $\smash{\Del_{\delta}(x')}$, where
  $\smash{x'=x_{[\dend+1:n-1]}}$.
Then we note that, for the algorithm to deviate from the invariant in the current round, one of the following two events must hold for $\smash{\by^1,\ldots,\by^N}$:\begin{flushleft}\begin{enumerate}
  \item $\DET(r,u,\{\by^1,\ldots,\by^N\})$ fails   
    \Cref{thm:DET}; or
  \item $\{\tilde{\bz}^1,\ldots,\tilde{\bz}^N\}$ 
   fails \Cref{thm:realBMA}, i.e., there is a multiset 
   $\smash{\{ z^1,\ldots, z^N\}}$ that is $(1/10)$-close to $\{\tilde{\bz}^1,\ldots,\tilde{\bz}^N\}$
   but $\smash{\BMA(n-\dend-1,\{z^1,\ldots,z^N\})}$ violates the condition in 
   \Cref{thm:realBMA}.
\end{enumerate}\end{flushleft}
This is because whenever $\DET$ succeeds, the strings $\{\bz^1,\ldots,\bz^N\}$ on which
  we run $\BMA$ on line 13 must be $(1/10)$-close to $\smash{\{\tilde{\bz}^1,\ldots,\tilde{\bz}^N\}}$.
\Cref{thm:DET} ensures that item~1  happens with probability at most $\smash{1/n^2}$;
\Cref{thm:realBMA} ensures that item~2 happens with probability at most $1/n^{2}$, 
  given that $\tilde{\bz}^1,\ldots,\tilde{\bz}^N$ are independent traces from $\Del_\delta(x')$
  as required in the assumption of \Cref{thm:realBMA}.

By a union bound, the invariant holds with high probability in every round given that we only repeat for $n/m$ rounds. Finally, observe that we only need to repeat for $n/m$ rounds to reconstruct the entire $n$-bit string $x$, since in each round the pointer $r$ increases by at least $2m$.

This concludes the proof of correctness of {\tt Reconstruct} and the proof of \Cref{thm:main}, modulo the proofs of \Cref{thm:realBMA} and \Cref{thm:DET}. In the rest of the paper we prove those two theorems.


\def\sfleft{\textup{\textsf{left}}}
\def\sfright{\textup{\textsf{right}}}

\section{Improved analysis of the Bitwise Majority Algorithm:  Proof\\
 of \Cref{thm:realBMA}}
\label{sec:bma}

\ignore{

}

The bitwise majority algorithm was first described and analyzed in \cite{BKKM04}. The analysis given in \cite{BKKM04} established that \BMA successfully reconstructs any
  unknown source string $x\in \{0,1\}^n$ that does not contain any ``long runs'' (i.e., subwords of the form $\smash{0^{n^{1/2 + \eps}}}$ or $\smash{1^{n^{1/2 + \eps}}}$) provided that the deletion rate $\delta$ is at most $n^{-(1/2 + \eps)}$.
We describe the \BMA algorithm in~\Cref{figg:BMA}.
As the main result of  this section we establish an improved performance guarantee for \BMA. Our discussion and notation below reflects the fact that we will in general be running \BMA ``in the middle'' of a string $x$ for which we have already reconstructed a $(b+1)$-bit prefix of $x$ (this is why \Cref{thm:realBMA} is stated in terms of a source string $x'$ of length $n'\le n$, which should be thought of as a suffix of $x$).

\ignore{
Let $x' = (x'_0,\dots,x'_{n-b-1}) = x_{[b:n-1]} \in \{0,1\}^{n-b}$ be an arbitrary string of length $n'$, where $n' = n-b$ and $\BMA$ is invoked with parameter $n' = n-b$ (intuitively, \BMA is called in an attempt to reconstruct the suffix $x_{[b:n-1]}$ of $x$). In what follows, throughout this section we write $r$ to denote the first location in $[b:n-1]$ that is deep in a desert of $x$ (recall \Cref{def:desert}) if $x_{[b:n-1]}$ contains a desert,
  and we set $r=n-m$ as a default value if $x_{[b:n-1]}$ has no desert.
  Note that with this definition of $r$, it is guaranteed that there is no desert in $x_{[b:r+m-1]} = x_{[b:n-1]}$ (equivalently, there is no desert in $x'_{[0:r+m-b-1]}$).
  }

Our goal is to prove \Cref{thm:realBMA}, restated below.  

\medskip
\noindent {\bf Restatement of \Cref{thm:realBMA}.} \emph{Let $\delta=n^{-(1/3+\eps)}$ for some constant $\eps>0$.
Suppose $\tilde{\bz}^{1},\ldots,\tilde{\bz}^{N}$ are $N$ independent traces drawn from $\Del_{\delta}(x')$ for some unknown string $\smash{x' = (x'_0,\dots,x'_{n'-1}) \in \{0,1\}^{n'}}$ with $n' \leq n$.
\hspace{-0.1cm}The following holds with probability at least $1-1/n^2$ over the randomness of $\smash{\tilde{\bz}^1,\ldots,\tilde{\bz}^\overallsc}$: }
  \begin{flushleft}\begin{enumerate}
    \item \emph{If $x'$ has no desert, then  $\BMA$ running on 
      $n'$ and any multiset $\{z^1,\ldots,z^\overallsc\}$ that is $(1/10)$-close to %
  $\{\tilde{\bz}^1,\ldots,\tilde{\bz}^\overallsc\}$ returns $w=x'$; 
  }
    \item \emph{If $x'$ has at least one desert, then 
    $\BMA$ running on 
      $n'$ and any multiset $\{z^1,\ldots,z^\overallsc\}$ that \\is $(1/10)$-close to %
  $\{\tilde{\bz}^1,\ldots,\tilde{\bz}^\overallsc\}$ returns a string $w$ that shares the 
    same $(r'+m+1)$-bit\\ prefix with $x'$, where $r'$ is the first location that is deep in a 
    desert in $x'$.
    } 
  \end{enumerate}\end{flushleft}

\begin{figure}[t]
  \centering
\setstretch{1.2}
  \begin{algorithm}[H]
    \caption{Algorithm {\BMA}}\label{fig:BMA}
		\DontPrintSemicolon
		\SetNoFillComment
		\KwIn{A length $n'$ and a multiset\ignore{\rnote{Changed some names here too; in keeping with {\tt Reconstruct}, the input traces are now called $z^1,\dots$}} $\{z^1,\ldots,z^{\overallsc}\}$ of strings, each of length
		at most $n'$}
		\KwOut{A string $w=(w_0,\dots,w_{n'-1})\in \{0,1\}^{n'}$}
		For each $i\in [\overallsc]$ pad each $z^i$ to be a string $u^i$ of length $n'$ by adding $0$'s to the end\\
		Set $t=0$ and $\current_i(t)=0$ for each $i\in [\overallsc]$\\
\While{$t\le n'-1$}{
Set $w_t \in \zo$ to be the majority of the $\overallsc$ bits $u^1_{\current_1(t)}, \ldots, u^\overallsc_{\current_\overallsc(t)}$\\
For each $i\in [\overallsc]$, set $\current_i(t+1)$ to $\current_i(t) + 1$ if $u^i_{\current_i(t)} = w_t$;
  \newline \ \ \ \hspace{0.5cm} otherwise set $\current_i(t+1)$ to $\current_i(t)$\\
Increment $t$.
}  
\Return $w$.
\end{algorithm}\caption{The Algorithm $\BMA$}\label{figg:BMA}
\end{figure}

We break the proof of \Cref{thm:realBMA} into two steps 
  (\Cref{lem:secondBMA} and \Cref{lem:firstBMA} below).
  For ease of exposition, in the rest of this section if $x'$ has at least one desert then as stated in item (2) of the theorem, we let $r'$ be the first location that is deep in a desert in $x'$.  If $x'$ has no desert, then we let $r' = n'-m-1.$
  Note that with this definition of $r'$, it is guaranteed that there is no desert in $\smash{x'_{[0:r'+m-1]}}$ and the goal of \BMA is to return a string
    that shares the same $(r'+m+1)$-prefix with $x$.

Let $R=9\overallsc/10$.
We first prove in \Cref{lem:secondBMA} that if a multiset
  of $R$ traces $Z=\{z^1,\ldots,z^R\}$~of $x'$ satisfies a certain sufficient ``goodness'' condition (see \Cref{def:goodset} for details),
  then $\BMA(n',Z)$~not only returns a string 
 $w=(w_0,\dots,w_{n'-1})\in \{0,1\}^n$ that satisfies
  $w_{[0:r'+m]}=x'_{[0:r'+m]}$ as desired~but moreover, the bitwise majority during each of the first $\smash{r'+m+1}$ rounds of \BMA is ``robust'' in the following sense:  for each one of those rounds,
  at least $9R/10 = 81N/100$ of the $R$ strings $z^i$'s agree with each other.
This immediately implies that when $Z$ satisfies this condition, 
 {adding any multiset of $\overallsc/10$ strings to $Z$ 
 and running $\BMA$ on the resulting multiset of size $N$} cannot affect the output of $\BMA$ during the first $r'+m+1$ rounds, so its output $w$
  still satisfies $w_{[0:r'+m]}=x'_{[0: r'+m]}$.
\ignore{
}The next lemma, \Cref{lem:firstBMA}, shows that if $\smash{\tilde{\bZ}=\{\tilde{\bz}^1,\ldots,\tilde{\bz}^\overallsc\}}$
  is a multiset of $\overallsc$ traces drawn independently from $\Del_\delta(x')$ (as in
  the assumption part of \Cref{thm:realBMA}), then with high probability  
  \emph{every} $R$-element subset of $\tilde{\bZ}$ satisfies the sufficient condition
  (\Cref{def:goodset})
  for $\BMA$ to succeed robustly.
 \Cref{thm:realBMA}  follows easily by combining \Cref{lem:secondBMA} and \Cref{lem:firstBMA}.
  
\subsection{Notation for traces}

We start with some useful notation for analyzing traces of $x'$.
When a trace $\by$ is drawn from $\Del_\delta(x')$ 
  we write $\bD$ to denote the set of \emph{locations 
  deleted} when $x'$ goes through the deletion channel, i.e., 
  $\bD$ is obtained by including each
  element in $[0:n'-1]$ independently with probability $\delta$, and
  $\by$ is set to be $\smash{x'_{[0:n'-1]\setminus \bD}}$.
In the analysis of  \BMA  when it is given as input
  $R$ traces $\smash{Z=\{z^1,\ldots,z^R\}}$,~our analysis will sometimes refer to the set $D_i\subseteq [0:n'-1]$ of locations that was deleted when generating $z^i$.

Note that in the execution of \BMA we pad each trace $z^i$ to 
  a string $u^i$ of length $n'$ by adding 
  $0$'s to its end\ignore{, where $n-b$ is the first argument in the invocation of \BMA (when it is invoked in Line~2 of {\tt Reconstruct} this corresponds to $b=0$) which gives the length of the remaining portion of $x$ that is yet to be reconstructed}.
In the rest of the section it will be convenient for us to view $x'$ as
  a string of infinite length by adding infinitely many $0$'s to its end. We can then view each $u^i$ as generated by first
  deleting the bits in $D_i\subseteq [0:n'-1]$ from $x' $
  and taking the $n'$-bit prefix of what remains. 
This motivates the definition of the following map $f_i:[0:n'-1]\rightarrow \mathbb{N}$
  for each $i\in [R]$:
For each $j\in [0:n'-1]$, $f_i(j)$ is set to be the unique integer $k$ such that 
  $k\notin D_i$ and $k-|D_i\cap [k-1]|=j$. 
In words, $f_i(j)$ is simply the original location in $x'$ of the $j$-th bit in 
  the padded version $u^i$ of $z^i$. 
  
  We specify some parameters that will be used in the rest of \Cref{sec:bma}.
  Let $C=\lceil 100/\eps\rceil$ (so $C$ should be thought of as a large absolute constant) and $M=2m+1$ with $m = n^{1/3}$, and recall that by definition $M$ is the shortest possible length of a desert.

\subsection{\BMA is robust on good sets of traces} 

The main result of this subsection is \Cref{lem:secondBMA}, which establishes that \BMA is robustly correct in its operation on traces that satisfy a particular ``goodness'' condition given in \Cref{def:goodset} below.

Let $Z=\{z^1,\ldots,z^R\}$ be a multiset of traces of $x'$.
As described above we write $u^i\in \{0,1\}^{n'}$ to denote the 0-padded version of $z^i$, $D_i\subseteq [0:n'-1]$ to denote the
  set of locations that were deleted from $x'$ to form $z^i$, and $f_i$ to denote the map 
    defined as above for each $i\in [R]$.
We introduce the following condition for $Z$ and 
  then prove \Cref{lem:secondBMA}: 

\begin{definition}\label{def:goodset}
We say  $Z=\{z^1,\ldots,z^R\}$ is \emph{good} if the following
  two conditions hold:
\begin{flushleft}\begin{enumerate}
\item[(i)] For every $i\in [R]$ and every interval $[\sfleft:\sfright]\subset [0:n'-1]$ of length 
  $\sfright-\sfleft+1=L_1:=2C^2 M$, we have 
  $ |D_i\cap [\sfleft:\sfright] |\le C.$ 
\item[(ii)] For every interval $[\sfleft:\sfright]\subset [0:n'-1]$ of length $\sfright-\sfleft+1=L_2:=M+C+1$,
  the number of elements $i\in [R]$ such that $ D_i\cap [\sfleft:\sfright] \ne \emptyset$ is at most $R/C^3$.
\end{enumerate}\end{flushleft}
\end{definition}

Intuitively, (i) says that no interval of moderate length (note that this length $2C^2 M$ is polynomially less than $1/\delta$) has ``too many'' deletions in it in any trace, whereas (ii) says that for every interval of moderate length (again polynomially less than $1/\delta$), most of the $R$ traces have no bit deleted within that interval.

Now we are ready to state \Cref{lem:secondBMA}:

\begin{lemma}\label{lem:secondBMA}
Let $Z=\{z^1,\ldots,z^R\}$ be a good multiset of $R$ traces of $x'$.
Then the string $w\in \{0,1\}^{n'}$ that $\emph{\BMA}(n',Z)$ outputs satisfies $\smash{w_{[0:r'+m]}=x'_{[0:r'+m]}}$. 
Moreover,  during 
  each of the first $r'+m+1$ rounds of the execution of $\emph{\BMA}$, at least $9R/10$ of the $R$ bits in the majority vote taken in Step~4 of $\emph{\BMA}$ agree with each other.
\end{lemma}

We start the proof of \Cref{lem:secondBMA} by defining a map 
  $\distance_i(t)$ for each $z^i$ in $Z$.
Recall that $\current_i(t)$ is the current location of the pointer into the padded trace $u^i$ at the beginning of round $t$ in \BMA.\footnote{Note that whereas $\position_i(\cdot)$ and $\distance_i(\cdot)$ refer to quantities defined in terms of the source string $x$, $\current_i(\cdot)$ refers to a location in a trace string and not the source string.}
We let $\position_i(t)=f_i(\current_i(t))$, i.e. the original position in $x'$ of the $\current_i(t)$-th bit
  of $u^i$.
Then $\distance_i(t)$ is defined as $\distance_i(t)=\position_i(t)-t$, the distance between $t$ and $\position_i(t).$
In \Cref{cor:nonneg} we will show that $\distance_i(t)$ is always nonnegative, and so it actually measures how many bits $\position_i(t)$ is ahead at round $t$.
It may be helpful to visualize $t$ and $\position_i(t)$ of a trace by writing down the source string $x'$ with the deleted bits struck through, and having two arrows pointing to $x'_t$ and $x'_{\position_i(t)}$ (see \Cref{fig:pointers}); at the beginning of round $t$, the \BMA algorithm tries to determine $x'_t$ by looking at $x'_{\position_i(t)}$.
Intuitively, having $\distance_i(t)=0$ means that the $i$-th trace was aligned properly at round $t$; at the highest level, we establish \Cref{lem:secondBMA} by showing that at least $9R/10$ of the $R$ traces have $\distance_i(t)=0$.

\begin{figure}
\centering
\begin{center}
\begin{tabular}{ |c|c|c|c|c|c|c|c|c|c |  } 
 \hline
 Time $t$ & $x'_0$ & $x'_1$  & $x'_2$ & $x'_3$ & $x'_4$ & $x'_5$ & $x'_6$ & $x'_7$ & $\cdots$ \\
 \hline \hline
 $t=0$ & $\Downarrow$ & & & & & & & & \\
& 0 & 1  & \cancel{1} & \cancel{0}  & 0 & 0 & \cancel{0} & 1 & $\cdots$ \\ 
 $\position_i(t)=0$  & $\Uparrow$ & & & & & & & & \\ \hline
 $t=1$  & & $\Downarrow$ & & & & & & & \\
 & 0 & 1  & \cancel{1} & \cancel{0}  & 0 & 0 & \cancel{0} & 1 & $\cdots$ \\ 
 $\position_i(t)=1$   & & $\Uparrow$ & & & & & & & \\ \hline
 $t=2$  & & & $\Downarrow$ & & & & & & \\
 & 0 & 1  & \cancel{1} & \cancel{0}  & 0 & 0 & \cancel{0} & 1 & $\cdots$ \\ 
  $\position_i(t)=4$ & &  & & & $\Uparrow$ & & & & \\ \hline
 $t=3$  & & & & $\Downarrow$ & & & & & \\
 & 0 & 1  & \cancel{1} & \cancel{0}  & 0 & 0 & \cancel{0} & 1 & $\cdots$ \\ 
   $\position_i(t)=4$& &  & & & $\Uparrow$ & & & & \\ \hline
\end{tabular}
\end{center}
\caption{An illustration of the progress of \BMA on a trace $z^i$ of $x' = 01100001 \cdots$. Bits that are \cancel{struck through} are deleted and do not occur in the trace; thus the trace $z^i$ in this example is $z^i =  0 1 0 0 1 \cdots$. 
In each row the downward arrow $\Downarrow$ shows the location of $t$ and the upward arrow $\Uparrow$ shows the position of $\position_i(t)$. In this example the bit $x'_{\position_i(0)}=0$ pointed to at time $t=0$ correctly matches $x'_0=0$ and the bit $x'_{\position_i(1)}=1$ pointed to at time $t=1$ correctly matches $x'_1=1$, so at the end of each of these time steps, $\position_i(\cdot)$ correctly advances to the next bit of the trace (the next bit of $x'$ that was not deleted). The bit $x'_{\position_i(2)} = 0$ pointed to at time $t=2$ does not correctly match $x'_2=1,$ so at the end of time step $t=2$, $\position_i(\cdot)$ does not advance.} 
\label{fig:pointers}
\end{figure}


We prove the following claim about how $\current_i(t),\position_i(t)$ and $\distance_i(t)$
  compare to their values at the beginning of round $t-1$,
  assuming that the prefix $w_{[0:t-1]}$ of the output thus far matches $x'_{[0:t-1]}$.
  \begin{claim}\label{simpleclaim}
  Let $t$ be a positive integer such that $w_{[0:t-1]}=x'_{[0:t-1]}$.
    For each $i\in [R]$, we have 
    \begin{flushleft}\begin{enumerate}
     \item If $x'_{\position_i(t-1)}\ne x'_{t-1}$, then $\current_i(t)=\current_i(t-1)$,
        $\position_i(t)=\position_i(t-1)$ and $\distance_i(t)=\distance_i(t-1)-1$.
     \item If $x'_{\position_i(t-1)}=x'_{t-1}$, then $\current_i(t)=\current_i(t-1)+1$,
        $\position_i(t)=\position_i(t-1)+\ell+1$ and $\distance_i(t)=\distance_i(t-1)+\ell$, where 
        $\ell$ is the nonnegative integer such that
        $\position_i(t-1)+1, \ldots, \position_i(t-1)+\ell\in D_i$ and
        $\position_i(t-1)+\ell+1\notin D_i$
        \emph{(}or equivalently, $\ell=f_i(\current_i(t))-f_i(\current_i(t-1))-1$\emph{)}.
    \end{enumerate}\end{flushleft}
  \end{claim}
  \begin{proof}
    If $x'_{\position_i(t-1)} \ne x'_{t-1}$, then $\current_i$ does not move and points to the same bit in $z^i$, which must come from the same bit in $x'$.
    
    If $x'_{\position_i(t-1)} = x'_{t-1}$, then $\current_i$ points to the next bit in $z^i$, which comes from the next undeleted bit in $x'$.
    So $\position_i(t)$ points to the position of the first undeleted bit of $x'$ after $x_{\position_i(t-1)}$ which is $\position_i(t-1)+\ell+1$, if $\position_i(t-1)+1, \ldots, \position_i(t-1)+\ell \in D_i$ and $\position_i(t-1)+\ell+1 \not\in D_i$.
  \end{proof}
 
We have the following useful corollary of \Cref{simpleclaim}, which tells us that if $w_{[0:t-1]}=x'_{[0:t-1]}$ then each $\distance_i(t) \ge 0$ (in other words, no trace can have ``gotten behind'' where it should be):

\begin{corollary} \label{cor:nonneg}
Let $t$ be a positive integer such that $w_{[0:t-1]}=x'_{[0:t-1]}$.
Then $\distance_i(t)\ge 0$ for all $i\in [R]$.
\end{corollary} 
\begin{proof}
Fixing an $i\in [R]$, we prove by induction that $\distance_i(t')\ge 0$ for every $t'=0,1,2,\dots,t$.
The base case~when $t'=0$ is trivial.
Now assuming that $\distance_i(t'-1)\ge 0$ for some $t'\le t$, we show that $\distance_i(t')\ge 0$ as well.
The case when $\distance_i(t'-1)>0$ is trivial since it follows from \Cref{simpleclaim} that
  $\distance_i(t')$ can go down from $\distance_i(t'-1)$ by at most one.
On the other hand, if $\distance_i(t'-1)=0$ and thus, $\position_i(t'-1)=t'-1$, we are in the 
  second case of \Cref{simpleclaim} so $\distance_i(t')\ge \distance_i(t'-1)\ge 0$.
This finishes the induction.
\end{proof}

We prove three preliminary lemmas before proving \Cref{lem:secondBMA}.
Recall that $M=2m+1$ is the shortest possible length of a desert.
Assuming $w_{[0:t-1]}=x'_{[0:t-1]}$ for some $t>M$,  the first lemma shows that if $\distance_i(t-M)=0$ and no location of $x'$ is deleted between $t-M+1$ and $t$, then $\distance_i(t)$ must stay at $0$.
(Note that this lemma holds for general $M$ but we state it using $M=2m+1$
for convenience since this is how it will be used later.) Intuitively, this says that if a length-$M$ subword of $x'$ experiences no deletions, then a trace that is correctly aligned at the start of the subword will stay correctly aligned throughout the subword
 and at the end of the subword.

\begin{lemma}\label{simplelemma2}
Suppose that $w_{[0:t-1]}=x'_{[0:t-1]}$ for some $t>M$. 
Suppose that $i\in [R]$ is such that $ \distance_i(t-M)=0$ and $$D_i\cap \big[\position_i(t-M)+1:\position_i(t-M)+M\big ]=D_i\cap \big[t-M:t\big ]=\emptyset.$$ Then 
  we have $\distance_i(t)=0$.
\end{lemma}
\begin{proof}
This follows easily from repeated applications of the second part of \Cref{simpleclaim}.
\end{proof}

In the second lemma, we assume $t$ is such that $M<t\le r'+m +1$ by the choice of $r$.
We further assume that $w_{[0:t-1]}=x'_{[0:t-1]}$ and 
  $0<\distance_i(t-M)\le C$ for some $i\in [R]$.
We show that under these assumptions, if the subword of length $M$ in $x'$ starting at $\position_i(t-M)+1$ has no deletion, then $\distance_i(t)<\distance_i(t-M)$.
Intuitively, this says that prior to a desert, if the length-$M$ subword of $x'$ experiences no deletions and the alignment of a trace is only modestly ahead of where it should be at the start of the subword, then the alignment will improve by the end of the subword.
\begin{lemma}\label{simplelemma}
Let $M<t\le r'+m+1$ with $w_{[0:t-1]}=x'_{[0:t-1]}$. 
If $0<\distance_i(t-M)\le C$ for some $i\in [R]$ and
$$D_i\cap \big[\position_i(t-M)+1:\position_i(t-M)+M \big]=\emptyset,$$ then we have $\distance_i(t)<\distance_i(t-M)$.
\end{lemma}
\begin{proof}
Let $k=\distance_i(t-M)$ with $0<k\le C$.
Assume for contradiction that $\distance_i(t)\ge k$, and let us consider the value
of $\current_i(t)$ vis-a-vis $\current_i(t-M)$.  
Since the pointer into the $i$-th trace moves forward by at most one position
in each round, we have that $\current_i(t)\le \current_i(t-M)+M$.
Having $\current_i(t)<\current_i(t-M)+M$ would imply $\position_i(t)< \position_i(t-M)+M$
  given that there is no deletion in the subword, and thus we would have
  $\distance_i(t)<\distance_i(t-M)$, which violates $\distance_i(t) \ge k$.  Thus it must be the case that
  $\current_i(t)=\current_i(t-M)+M$.
  
In order to have $\current_i(t)=\current_i(t-M)+M$,
  we must have $$\current_i(t-M+\ell)=\current_i(t-M)+\ell$$ for every $\ell\in [M]$ (again because the 
    pointer moves forward by at most one each round).
By the second part of \Cref{simpleclaim} and the assumption that $w_{[0:t-1]}=x'_{[0:t-1]}$, this implies that
$$x'_{t-M+\ell-1}=x'_{\position_i(t-M+\ell-1)}=x'_{t-M+\ell-1+k}$$ for all $\ell\in [M]$.
Thus, the substring of $x'$ starting at $t-M$ and ending 
  at $t-1$ is a desert  of length $M$ with a pattern of length $k\le C$.  However, this contradicts with the choice of $r$ and the fact that $t\le r'+m+1$; this contradiction concludes the proof.
\end{proof}

Finally we use the two previous lemmas to show that if $t\le r'+m +1$ and $w_{[0:t-1]}=x'_{[0:t-1]}$, then
  $\distance_i(t)$ must lie between $0$ and $C$.
  Intuitively, this says that prior to a desert, the alignment of a trace will be at worst modestly ahead of where it should be.
  
\begin{lemma}\label{simplelemma3}
Let $t\le r'+m + 1$ and suppose that $w_{[0:t-1]}=x'_{[0:t-1]}$. 
Then $ \distance_i(t)\le C$ for all $i\in [R]$.
\end{lemma}
\begin{proof}
Assuming for a contradiction that $\distance_i(t)>C$, we write  $t^*\le t$ to denote the smallest 
  integer such that $\distance_i(t^*)>C$ (so $\distance_i(t')\le C$ for all $t'<t^*$).
  First we claim that we must have $\position_i(t^*) > L_1$ (recall that $L_1 = 2C^2 M$).
  By part (i) of \Cref{def:goodset}, there are no more than $C$ deletions in $D_i$ within the interval of $[0:L_1-1]$, and since by \Cref{simpleclaim} $\distance_i(\cdot)$ can only increase when a deletion occurs, this
  would mean that we have $\distance_i(t')\le C$ for $\position_i(t')\in [0,L_1]$.
  So we assume below that $\position_i(t^*) > L_1$.

To derive a contradiction,
  we consider the interval $$W=\big[\position_i(t^*)-L_1:\position_i(t^*)-1\big]$$ of length $L_1=2C^2 M$.
  By the first condition in \Cref{def:goodset}, we have that 
  $|D_i\cap W|\le C$.
  So it follows from the pigeonhole principle and the choice of $L_1$ that there must be a interval $[\sfleft:\sfright]$ of length $CM$ (so $\sfright=\sfleft+CM-1$)  
  inside $W$ that is disjoint from $D_i$.
Let $t'$ be the integer such that $\position_i(t')=\sfleft$.
Since $\position_i(t') < \position_i(t^*)$, we have $t' < t^*$, and hence $\distance_i(t')\le C$ by the choice of $t^*$.
We also have 
$[\position_i(t'+\ell M):\position_i(t'+ \ell M)+M-1]\cap D_i=\emptyset$ for each $\ell\in[0:C-1]$.
As a result, for each $\ell\in [0:C-1]$, either (by \Cref{simplelemma2}) 
  $\distance_i(t'+ \ell M)=0$ and $\distance_i(t'+(\ell+1) M)$ stays at $0$,
  or (by \Cref{simplelemma}) $\distance_i(t'+ \ell M)>0$ and $\distance_i(t'+(\ell+1) M)$ strictly decreases.
This implies $\distance_i(t'+CM)=0$. 
Note that $\position_i(t'+CM)\le \position_i(t')+CM \le \position_i(t^*)$.
It follows that $\distance_i(t^*)\le C$ since there are at most $C$ deletions $[\position_i(t'+CM):\position_i(t^*)-1]  \subseteq [\sfleft:\position_i(t^*)-1]\subseteq W$,
a contradiction.
\end{proof}

We are now ready to prove \Cref{lem:secondBMA}:

\begin{proof}[Proof of \Cref{lem:secondBMA}]
We prove by induction that for every positive integer $t\le r'+m+1$:
\begin{equation} \label{eq:inductive-statement}
w_{[0:t-1]}=x'_{[0:t-1]}\quad\text{and}\quad \sum_{i\in [R]} \distance_i(t)\le \frac{2R}{C}.
\end{equation}
It follows that every $t\le r'+m+1$ satisfies 
  $|\{i\in [R]:\distance_i(t)=0\}|\ge R-2R/C\ge 9R/10$ using $C\ge 20$.

The base case of the induction is when $t\le M $. This case is easy since by the
  second condition of \Cref{def:goodset}, the number of $i\in [R]$ such that 
  $D_i\cap [0:M]\ne \emptyset$ is at most $R/C^3$.
As a result,  we have 
  $w_{t'}=x'_{t'}$ for every $t'\in [0:M]$. It also follows from  \Cref{simplelemma3} that
  $\distance_i(t')\le C$ for all $t' \in [0:M]$ and $i\in [R]$ and thus,
$$\sum_{i\in [R]} \distance_i(t')\le 
  \frac{R}{C^3}\cdot C=\frac{R}{C^2}<\frac{2R}{C},$$ 
  which establishes the base case.

For the inductive step we consider $t$ such that $M+1\le t\le r'+m+1$.
Given $$w_{[t-2]}=x'_{[t-2]}\quad \text{and} \quad\sum_{i\in [R]} \distance_i(t-1)\le \frac{2R}{C}$$
by the inductive hypothesis,
it follows that the majority of $i\in [R]$ have $\distance_i(t-1)=0$. Hence
  $w_{t-1}=x'_{t-1}$ and consequently $w_{[0:t-1]}=x'_{[0:t-1]}$.
To conclude the proof it remains to bound $\sum_{i \in [R]} \distance_i(t)$ and this is done by comparing
  $\distance_i(t)$ with $\distance_i(t-M)$ for each $i\in [R]$.
Note that by \Cref{simplelemma3} we have that $0\le \distance_i(t-M)\le C$ for every $i \in [R]$, and hence we have that $$\big[\position_i(t-M):\position_i(t-M)+M\big]\subseteq
  \big[t-M:t+C\big]$$
  for every $i\in [R]$.  Since $[t-M:t+C]$ is a fixed interval of length $L_2$,
it follows from the second condition in \Cref{def:goodset} that the number of $i\in [R]$
  that have at least one deletion in the interval $[t-M:t+C]$ is
  at most $R/C^3$.
We consider the following three cases:
\begin{enumerate}
\item If $\distance_i(t-M)=0$ and $D_i\cap \big[t-M:t+C\big]=\emptyset$, then (by \Cref{simplelemma2}) we have that $\distance_i(t)=\distance_i(t-M)=0$;
\item If $\distance_i(t-M)>0$ and $D_i\cap \big[t-M:t+C\big]=\emptyset$, then (by \Cref{simplelemma}) we have that $\distance_i(t)\le \distance_i(t-M)-1$;
\item  If $D_i\cap \big[t-M:t+C\big]\ne \emptyset$, then we have that $\distance_i(t)\le C\le \distance_i(t-M)+ C$, where the first inequality is by \Cref{simplelemma3} and the second is by \Cref{cor:nonneg}.
\end{enumerate}
Recall that by item (ii) of \Cref{def:goodset}, at most $R/C^3$ of the $i \in [R]$ can have $D_i\cap \big[t-M:t+C\big]\ne \emptyset$. It follows from this and the above three cases that
$$
\sum_{i\in [R]} \distance_i(t)\le \sum_{i\in [R]} \distance_i(t-M)
+\frac{R}{C^3}\cdot C-\max\left\{0,\big|\{i\in [R]:\distance_i(t-M)>0\}\big|-\frac{R}{C^3}\right\}.
$$
If the number of $i\in [R]$ with $\distance_i(t-M)>0$ is no larger than $R/C^3$, then since by \Cref{simplelemma3} each such value of $\distance_i(t-M)$ is at most $C$, we get that
$$
\sum_{i\in [R]}\distance_i(t)\le \frac{R}{C^3}\cdot C+\frac{R}{C^3}\cdot C<\frac{2R}{C}.
$$
Otherwise there are at least $R/C^3$ many elements $i \in [R]$ with $\distance_i(t-M)>0.$ In this case, using (again by \Cref{simplelemma3}) the bound  $$\big|\{i\in [R]: \distance_i(t-M)>0\}\big|\ge  \frac{1}{C}\sum_{i\in [R]} \distance_i(t-M),$$ 
it follows from the inductive hypothesis that 
$$
\sum_{i\in [R]}\distance_i(t)\le 
\left(1-\frac{1}{C}\right)\sum_{i\in [R]} \distance_i(t-M)+\frac{R}{C^2}+\frac{R}{C^3}
\le \left(1-\frac{1}{C}\right)\frac{2R}{C}+\frac{R}{C^2}+\frac{R}{C^3}
\le \frac{2R}{C}.
$$
This finishes the induction and the proof of \Cref{lem:secondBMA}.
\end{proof}

\subsection{Traces are good with high probability}

To conclude the proof of \Cref{thm:realBMA} it remains to prove \Cref{lem:firstBMA}, which states that with high probability a random multiset of 
$O(\log n)$ traces is such that every subset of $9/10$ of the traces is good
  (recall \Cref{def:goodset}):

\begin{lemma}\label{lem:firstBMA}
Let $\tilde{\bZ}=\{\tilde{\bz}^1,\ldots,\tilde{\bz}^\overallsc\}$ be a multiset of $\overallsc=O(\log n)$ traces drawn independently from
  $\Del_\delta(x')$.
Then with probability at least $1-1/n^2$, every $R$-subset of $\tilde{\bZ}$ is good, where
  $R=9\overallsc/10$.
\end{lemma}
\begin{proof}
It suffices to show that with probability at least $1-1/n^2$:
\begin{flushleft}\begin{enumerate}
\item[(i)] For every $i\in [\overallsc]$ and every interval $[\ell:r]\subset  [0:n'-1]$ of length 
  $r-\ell+1=L_1:=2C^2 M$, we have 
   $ |D_i\cap [\ell:r] |\le C.$ 
\item[(ii)] For every interval $[\ell:r]\subset  [0:n'-1]$ of length $r-\ell+1=L_2:=M+C+1$,
  the number of $i\in [\overallsc]$ such that $ D_i\cap [\ell:r] \ne \emptyset$ is at most $R/C^3$.
\end{enumerate}\end{flushleft}
We start with (i). For any interval of length $L_1$ and any $i \in [R]$,
  the probability of $ |\bD_i\cap [\ell:r]|> C$ is at most
  $$
 {L_1\choose C+1} \delta^{C+1} \le (L_1\delta)^{C+1} \le 
 (2C^2M\delta)^{C+1} =O_\eps(1) \cdot n^{-(C+1)\eps}\le O(1/n^4),
$$
where we used $M = \Theta(n^{1/3})$ and $C=\lceil 100/\eps\rceil$.
It then follows from a union bound over intervals of length $L_1$ and $i \in [R]$ that (i) is violated with probability at most
  $O(1/n^4)\cdot n\overallsc=o(1/n^2)$. 
  
  For (ii), note for each interval of length $L_2$
    and each $i\in [\overallsc]$, the probability of $\bD_i\cap [\ell:r]\ne \emptyset$ is
    at most $\delta L_2=O(n^{-\eps/3})$.
Similarly the probability of having at least $R/C^3$ such $i$ is at most
$$
{\overallsc\choose R/C^3} \left(\delta L_2\right)^{R/C^3} \le \left(\overallsc\delta L_2\right)^{R/C^3}
\le n^{-\Omega(\log n)},$$
using $R=9\overallsc/10$.
\ignore{
}
It follows from a union bound over (at most $n$) intervals of length $L_2$ that (ii) is violated with  probability $o(1/n^2)$. The lemma follows from a union bound over (i) and (ii).
\end{proof}



\def\CEsc{2/\eps}
\def\SEsc{O(n^{2/3-\eps})}
\def\constsc{1000}
\def\overallsc{N}

\section{Finding the end of a desert:  Proof of \Cref{thm:DET}} \label{sec:long_runs}


In this section, we describe the algorithm \DET, which is used to determine the end of a desert
  in $x$ using traces from $\Del_\delta(x)$, and to align given traces
  with the end of the desert. (These aligned traces will then 
  be used by $\BMA$ in the main algorithm.)

Let's recall the setting. Let $x\in \{0,1\}^n$ be the unknown string.
\DET is given the first location $r$ that is deep in some $s$-desert subword of $x$,
  for some 
  string $s\in \{0,1\}^k$ with $k\le C$.
It is also given  
the  prefix $u=x_{[0:r+m]}$ of $x$.
We will refer to the $s$-desert that contains $r$ as the \emph{current} desert.
(Note that $s$ can be easily derived from $u$.)
The goal of $\DET$ is to figure out the ending location  
of the current desert which we denote by $\dend$:
\[
\dend
\text{\emph{~is the smallest integer at least $r+m$ 
  such that $x_{\dend+1}\ne x_{\dend-k+1}$}.}
\]
(Note that thanks to the preprocessing step {\tt Preprocess}, we know that 
  $\dend$ exists and satisfies $r+m\le  \dend\le 3n/4$.)  

In addition to computing $\dend$, $\DET$ is also given a multiset of
  $\overallsc=O(\log n)$ traces $y^1,\ldots,y^\overallsc$ and needs to return a location $\ell_i$ for each $y^i$
  such that most of them are correctly aligned to the end of the desert.
Formally, we write $\pos(y)$ for a trace $y$ to denote the location $\ell$
  in $y$ such that $y_\ell$ corresponds to the last bit of $x_{[0:\dend]}$ that survives in $y$;
  we set $\pos(y)=-1$ by default if all of $x_{[0:\dend]}$ gets deleted. 
The second goal of $\DET$ is to output $\ell_i=\pos(y^i)$ for almost all $y^i$
  when they are drawn independently from $\Del_\delta(x)$.

We  restate the main theorem of this section: 
\medskip

\noindent {\bf Restatement of \Cref{thm:DET}:}
\ignore{\rnote{emph interacts badly with itemize; I think you need to do separate emphs for each item, or at least that's a workaround}}
\emph{
  Let $\delta = n^{-(1/3+\eps)}$ for some fixed constant $\eps>0$. There is an algorithm $\DET$ with the following input and output:}
\begin{flushleft}
\begin{itemize}
\item  \emph{{\bf Input:} (i) a location $r \in [0:3n/4]$,
  (ii) a string $u\in \{0,1\}^{r+m+1}$, 
  (iii) a multiset of
  strings $\{y^1,\ldots,y^{\overallsc}\}$ from $\{0,1\}^{\le n}$ where $\overallsc=O(\log n)$, and (iv) sample access
  to $\Del_\delta(x)$ for some unknown $x\in \{0,1\}^n$.
}

\item \emph{{\bf Output:} An integer $b$, and an integer $\ell_i$ for each $i\in [\overallsc]$.}

\end{itemize}
\end{flushleft}
\emph{The algorithm \DET draws $\tilde{O}(n^{2/3-\eps})$ many independent 
	traces from $\Del_\delta(x)$, runs in $O(n^{5/3})$ time and
	has the following performance guarantee.
	Suppose that $r$ is the first location that is deep in~some desert of $x$; $u= x_{[0:r+m]}$; the unknown $\dend$ of the desert to which $x_r$ belongs
	is at most  $3n/4$; and $\smash{y^1=\by^1,\ldots,y^\overallsc=\by^{\overallsc}}$ 
	are independent traces drawn from $\Del_\delta(x)$. 
	Then the integers $b$ and $\ell_i$ that $\DET$ outputs 
	satisfy the following properties with probability
	at least $1-1/n^2$:
	$b=\dend$, and $\ell_i=\pos(\by^i)$ for at least $0.9$ fraction of $i\in [\overallsc]$.}\medskip 
We present the algorithm \DET in \Cref{fig:DET}, where 
 \[
\sigma:=\big{\lceil} \sqrt{\delta n} \cdot \log n\big{\rceil}. 
 \]
(Intuitively, $\sigma$ provides a high-probability upper bound on how far a bit of $x$ can deviate from
  its expected position in a trace $\by\sim \Del_\delta(x)$.)
\DET consists of the following two main procedures:
\begin{flushleft}
\begin{enumerate}
\item 
We will refer to the $8\sigma$-bit string 
 $$
 \tail := x_{\dend-k+2}\ x_{\dend-k+3}\ \cdots\  x_{\dend+8\sigma-k+1}$$
  around the end $x_{\dend}$ of the current desert as its \emph{tail} string and denote it 
  by $\tail\in \{0,1\}^{8\sigma}$.
(Note that $\dend+8\sigma-k+1<n$ given that $\dend\le 3n/4$.)
The first procedure, \CE, 
will provide with high probability a \emph{coarse estimate} $\hat{\beta}$ (see \Cref{lem:Coarse-Estimate}) of the expected location
  $(1-\delta)\dend$  
of the 
right end of the current desert in a trace of $x$. Moreover, it 
  returns a string $t\in \{0,1\}^{8\sigma}$ that is exactly the $\tail$ string with high probability.

\item 
With $\hat{\beta}$ and $\tail\in \{0,1\}^{8\sigma}$ in hand, 
  the second procedure \Align 
  can help align a given trace with the right end of the current desert. 
Informally, running on a trace $\by\sim \Del_\delta(x)$, \Align returns 
  a position $\ell$ such that  with high probability over the randomness of $\by\sim \Del_\delta(x)$, it holds that $\ell=\pos(\by)$.
The performance guarantee of \Align is given in \Cref{lem:Align}. 
It may sometimes (with a small probability) return $\nil$, meaning that it fails to align the given trace.
\end{enumerate}\end{flushleft} 

The algorithm \DET starts by running \CE to obtain a coarse estimate $\hat{\beta}$
  of $(1-\delta)\dend$ and the $\tail$ string (line~1).
It then (line~2) runs \Align on the given $\overallsc$ traces $y^i$ to obtain $\ell_i$ for each $i\in [\overallsc]$.
The second property of \DET in \Cref{thm:DET} about $\ell_i$'s follows
  directly from the performance guarantee of \Align. 
To obtain a sharp estimate of $\dend$, \DET draws another set of $\tilde{O}(n^{2/3-\eps})$ traces $\bz^i$ (line~3).
It runs \Align on each of them and uses the average of its outputs (discarding traces for which \Align returns \nil) to estimate
   $(1-\delta)\dend$ (lines~4-6). 
   (It is clear that this average would be accurate to within $\pm o(1)$ if \Align always successfully aligned its input trace with the right end of the current desert; the actual performance guarantee of \Align is weaker than this, but a careful analysis enables us to show that it is good enough for our purposes.)

The rest of this section is structured as follows. 
In \Cref{sec:CE} we describe the \CE procedure and establish a performance guarantee for it (\Cref{lem:Coarse-Estimate}), and in \Cref{sec:align} we describe the \Align procedure and establish a performance guarantee for it (\Cref{lem:Align}). We combine these ingredients to prove \Cref{thm:DET} in \Cref{sec:proof-of-det}.

\begin{figure}[t]
\setstretch{1.2}
  \begin{algorithm}[H]
    \caption{Algorithm {\DET}}\label{fig:DET}
		\DontPrintSemicolon
		\SetNoFillComment
		\KwIn{$r \in [0:3n/4]$,\vspace{-0.03cm} $u\in \{0,1\}^{r+m+1}$, 
		 a multiset $ \{ y^1,\ldots, y^{\overallsc}\}$ of $\overallsc$ strings
		from $\{0,1\}^{\le n}$ where $\overallsc = O(\log n)$, and sample access to $\smash{\Del_\delta(x)}$ for some string $x\in \{0,1\}^n$.}
		\KwOut{An integer $b\ge r+m$ and an integer $\ell_i$ for each $i\in [\overallsc]$.}
		Run {\tt Coarse-Estimate}$(r, u)$, which returns an integer $\hat{\beta}$ and 
		   a string $t\in \{0,1\}^{8\sigma}$.\\ 
		For each $i\in [\overallsc]$, run ${\tt Align}(\hat{\beta},t,y^i)$. 
		If ${\tt Align}$ returns \nil, set $\ell_i=-1$; otherwise let $\ell_i$ be the integer 
		 ${\tt Align}$ 
	      returns.\\
	    Draw $\gamma = O(n^{2/3-\eps} \log^3 n)$ 
	      traces $\bz^1,\ldots,\bz^\gamma$ from $\Del_\delta(x)$.\\
		For each $i\in [\gamma]$, run ${\tt Align}(\hat{\beta},t,\bz^i)$ 
		  and let $h_i$ be its output.\\
		Let $\beta$ be the average of $h_i$'s 
		that are not $\nil$, 
		and let $b$ be the 
		integer nearest to $\beta/(1-\delta)$.\\
 		Return $b$, and $\ell_i$ for each $i\in [\overallsc]$.
\end{algorithm}
\end{figure}

\subsection{The \CE procedure} 
\label{sec:CE}

Recall that $\sigma =\lceil \sqrt{\delta n} \cdot \log n\rceil.$ 
Given $r, u$ as specified earlier and sample access to $\Del_\delta(x)$,
  the~goal of \CE is to obtain an integer  $\hat{\beta}$ such that 
  $|\hat{\beta}-(1-\delta)\dend|\leq 2\sigma $.
We will refer to such an estimate as a \emph{coarse estimate} of $(1-\delta)\dend$. 
In addition, \CE returns a string $t$ that with high probability is exactly the tail string $\tail\in \{0,1\}^{8\sigma}$.
This is done by drawing only $O(1/\eps)$ many traces.

\begin{lemma} \label{lem:Coarse-Estimate}
  Let $\delta = n^{-(1/3+\eps)}$ with a fixed constant $\eps>0$. There is an algorithm \emph{\CE} which takes the same two inputs $r$ and $u$ as in $\emph{\DET}$ and sample access to $\Del_\delta(x)$ for some unknown string $x\in \{0,1\}^n$, and returns an integer $\hat{\beta} \in [0:n-1]$ and a string $t\in \{0,1\}^{8\sigma}$.
It draws $O(1/\eps)$ traces from $\Del_\delta(x)$, runs in time $O(n)$ and has the following performance guarantee.
Suppose that $r$ and $u$ satisfy the same conditions as in \Cref{thm:DET}
with respect to $x$.
Then with probability at least $1-1/n^3$, we have that $t=\emph{\tail}$ and 
  $\hat{\beta}$ satisfies 
  $|\hat{\beta} - (1-\delta)\emph{\dend}|\le 2\sigma$. 
\end{lemma}

\begin{proof}
We start with the coarse estimate $\hat{\beta}$.
Let $\hat{r}=\lceil (1-\delta)r\rceil$ and
  consider the following collection of overlapping intervals of positions in a trace of $x$:
\[
{\cal I} := \Big\{\big[\hat{r}+j\sigma:\hat{r}+(j+4)\sigma\big] : j \in \Z_{\geq 0}\Big\}.
\]
Note that each interval $I$ contains $4\sigma+1$ positions.
\CE  draws $\alpha=O(1/\eps)$ traces $\by^1,\ldots,\by^\alpha$ from $\Del_\delta(x)$
  and finds the leftmost interval $I^*\in {\cal I}$ such that at least
  half of $\by^i$'s satisfy the following property: $\by^i_{I^*}$ contains
  a $k$-bit subword that is not a cyclic shift of $s$.  
We then set $\hat{\beta}$ to be the right endpoint of $I^*$ and 
  show that it is a coarse estimate of $(1-\delta)\dend$ with high probability.

\def\lleft{\textnormal{left}}
\def\rright{\textnormal{right}}

For this purpose consider the first interval $I'=[\lleft',\rright'] \in \calI$ 
such that $(1-\delta)\dend\le  \rright'-\sigma$.
Such an interval $I'\in \calI$ must exist given that $\dend-r\ge m \gg \sigma$
  and the fact that intervals of $\calI$ are staggered at offsets of $\sigma$ from each other.
We also have $(1-\delta)\dend>\rright'-2\sigma$ given that $I'$ is chosen to be 
  the leftmost 
  such interval. So $(1-\delta)\dend$ is at least $2\sigma$ to the right of $\lleft'$.
For a trace $\by\sim \Del_\delta(x)$ we are interested in the event that (i) the $k$ bits
  $x_{\dend-k+2}\cdots x_{\dend+1}$ of $x$ all survive in $\by$ and (ii) they
  all lie in $\by_{I'}$.
This event occurs with probability at least $1-O(k\delta)$: for it to not occur,
  either one of these $k$ bits gets deleted, which happens with probability 
  at most $k\delta$, or they need to deviate from their expected location 
  in $\by$ by more than $\sigma$, which happens with even smaller probability (this holds by
  a Chernoff bound given our choice of $\sigma$).
When the event described above on $\by$  occurs, $\by_{I'}$ contains a 
  $k$-bit string that is not a cyclic shift of $s$. 
As a result, when we draw $\by^1,\ldots,\by^\alpha\sim \Del_\delta(x)$,
  the probability that no more than half of them satisfy this property is at most
$$
2^{\alpha}\cdot \big(O(k\delta)\big)^{\alpha/2}<1/n^5\ignore{\cnote{Let's change this to $1/n^5$ or something smaller?  If we look at two equations below, we have the same qualtity except with $k$ replaced by $\sigma \gg k$ and we get $1/n^4$.}}
$$
when $\alpha$ is a large enough constant.
 Hence with probability at least $1 - 1/n^5$,
   the interval $I^*$ picked by \CE is either $I'$ or  
some interval $I \in {\cal I}$ to the left of $I'$. 

\def\ccenter{\textsf{center}}

Next, we argue that $I^*$ is unlikely to be any interval $I=[\lleft:\rright]\in \calI$ 
with $\lleft\le \lleft'-3\sigma$ (so $\rright \le \rright' - 3\sigma$).
Fix any such interval $I=[\lleft:\rright]$.
We have 
$$(1-\delta)\dend> \rright'-2\sigma \ge \rright+\sigma$$
and thus, 
%
  $(1-\delta)\dend$ is to the right of $I$ by at least $\sigma$.
Let $\tilde{I}$ be the interval of positions of $x$ that starts with 
  $\lleft/(1-\delta)-\sigma$ and ends with $\rright/(1-\delta)+\sigma<\dend$.
For a trace $\by\sim \Del_\delta(x)$, we are interested in the event that (a)
all the bits in $\by_I$ come from $x_{\tilde{I}}$, and (b) no bit of $x_{\tilde{I}}$ is deleted (though of course they need not all end up in interval $I$ in $\by$).
It follows from an argument similar to the one given above that this occurs with probability
  at least $1-O(\sigma \delta)$; when it occurs, since $x_{\tilde{I}}$ is contained in the current $s$-desert, every $k$-bit string of 
  $\by_I$ is a cyclic shift of $s$.
As a result, for traces $\by^1,\ldots,\by^\alpha$, the probability that 
  at least half of them has a non-cyclic shift in $\by^i_I$ is at most
$$
2^{\alpha}\cdot \big(O(\sigma\delta)\big)^{\alpha/2}<1/n^4
$$
using $\alpha=O(1/\eps)$.
It follows from a union bound over the $O(n/\sigma)$ intervals
  that $I^*$ is an interval $I=[\lleft:\rright]\in \calI$ with $\lleft\le \lleft'-3\sigma$
  with probability at most $1/n^3$.

Combining these two parts,
  with probability at least $1-2/n^3$, we have
  $\rright'-2\sigma\le \hat{\beta}\le \rright'$.
Given that $\rright'-2\sigma\le (1-\delta)\dend\le \rright'-\sigma$,
  we have $|\hat{\beta}-(1-\delta)\dend|\le 2\sigma$ when this happens.  

Finally, given a coarse estimate $\hat{\beta}$ as above,
  \CE recovers the tail string as follows.
Let $J'$ be the interval $[\hat{\beta}-3\sigma: \hat{\beta}+3\sigma]$ of positions in
  a trace.
We draw another sequence of $\alpha=O(1/\eps)$ fresh traces $\by^1,\ldots,\by^\alpha$
  from $\Del_\delta(x)$. 
For each $\by^i$ we look for the leftmost non-cyclic shift of $s$ in $\by^i_{J'}$.
When such a non-cyclic shift exists, say $\by^i_{\tau}\cdots \by^i_{\tau+k-1}$,
  $\by^i$ \emph{votes} for the $8\sigma$-bit  string
  $\by^i_{\tau} \cdots \by^i_{\tau+8\sigma-1}$ as its candidate for the tail string.
\CE then returns the $8\sigma$-bit string with the highest number of votes.

Assuming that $\hat{\beta}$ satisfies $|\hat{\beta}-(1-\delta)\dend|\le 2\sigma$,
  we show that $\CE$ returns the tail string correctly with high probability.
To this end, as $\by\sim \Del_\delta(x)$, we are interested in the event   
that (1) no bit in $x_{[\dend-5\sigma:\dend+11\sigma]}$ is deleted and (2) $x_{\dend+1}$
  lies in the interval $J'$ in $\by$.
It is easy to argue that this event occurs with probability 
  at least $1-O(\sigma \delta)$, and when this occurs, 
  $\by$ votes for the correct tail string.
It then follows from a similar argument that, with $\by^1,\ldots,\by^\alpha\sim 
\Del_\delta(x)$, at least half of strings vote for the correct tail string 
  with probability at least $1-1/n^3$.
  
It follows from a union bound that with probability at least $1-O(1/n^3)$,
  $\hat{\beta}$ is a coarse estimate of $(1-\delta)\dend$
  and the string  returned by \CE is exactly the tail string. Clearly, the running time of \CE is $O(n)$ as the procedure consists of a linear scan over $O(1/\eps)$ traces. This finishes
  the proof of the lemma.
\end{proof}

\def\windowlength{10\sigma}
\def\DETthreshold{7\sigma}
\def\windowdistanced{3\sigma}
\def\constDeletion{1000/\eps}

\subsection{The \Align procedure}
\label{sec:align}

We start with the performance guarantee of $\Align$:

\begin{lemma} \label{lem:Align}
	Let $\delta = n^{-(1/3+\eps)}$ for some fixed constant $\eps>0$. There is an 
	algorithm \Align running in time $O(n)$ 
	with the following input and output:
	\begin{flushleft}
		\begin{itemize}
			\item  {\bf Input:} a number $\hat{\beta} \in [0:n-1]$, and strings $t \in \zo^{8\sigma}, y \in \zo^{\leq n}.$
			\item  {\bf Output:} an integer $\ell \in [0:n-1]$, or \nil.
		\end{itemize}
	\end{flushleft}
	It has the following performance guarantee. 
	Suppose $x,u,r$ and $\dend$ satisfy the hypothesis in \Cref{thm:DET}, 
	 $\hat{\beta}$ and $t$ satisfy the conclusion of \Cref{lem:Coarse-Estimate}, and $y=\by \sim \Del_\delta(x)$
	is a random trace.~Then 
		\begin{flushleft}
		\begin{enumerate}
		\item   
		Whenever $\Align$ returns an integer $\ell$, we have  
$ |\hspace{0.03cm}\ell - (1-\delta)\dend\hspace{0.03cm} |\leq O(\sigma)$.
			\item With probability at least $1-\tilde{O}(n^{-3\eps/2})$, $\Align$ returns
			exactly $\pos(\by)$; and
			\item Conditioned on $\Align$ not returning $\nil$, the expectation of what $\Align$ returns   
			   is $(1-\delta)\dend\pm o(1)$.
		\end{enumerate}
	\end{flushleft}
\end{lemma}

\subsubsection{Setup for the proof of \Cref{lem:Align}} \label{sec:setup}

The special case when $k=|s|=1$ (so the desert subword is of the form $0^a$ or $1^a$ for some $a \ge M = 2n^{1/3}+1$) is relatively simple and will be handled separately in \Cref{sec:kisone}, so in the rest of this subsection we set up ingredients needed for the proof of \Cref{lem:Align} in the general case when $k \ge 2$.

Consider the general case when $k\ge 2$. Let $\Cyc$ be the set of all $k$-bit strings that can be obtained as cyclic shifts of $s$. The key notion behind \Align is the idea of the ``\emph{signature}.'' 
This is a subword of $x$ of length at most $8\sigma$ that starts at the 
  same location $x_{\dend-k+2}$ as $\tail$ (so it is contained in $\tail$; we remind the reader that 
  the first $k$-bits of $\tail$ is a string not in $\Cyc$).
The signature ends at location $d$ where $d$ is the smallest integer $d\in [\dend+k+1:\dend+8\sigma-k+1]$
  such that the $k$-bit subword that ends at $d$ is not in $\Cyc$;
  if no such $d$ exists,  the signature is taken to have length $8 \sigma$ and is the same as $\tail$. 
(Alternatively, we can describe the signature as the shortest prefix of $\tail$
  that contains a $k$-bit subword not in $\Cyc$ that does not use the first $k$ bits; and 
  it is set to $\tail$ if every $k$-bit subword of
  $\tail$ after removing the first $k$ bits lies in $\Cyc$.)
  
We will write $\sig$ to denote the signature string, and we observe that as an immediate consequence of the definition of $\sig$ given above, we have 
  $2k\le |\sig|\le 8\sigma$. We further observe that given the string $\tail$ it is algorithmically straightforward to obtain $\sig$ by following the definition given above.

Given $\sig$, we say that a string $z$ of length at most $15\sigma+1$ is in the \emph{right} form if it can be written as
\begin{equation}\label{eq:rightform}
z=w\circ \sig
\end{equation}
where the leftmost $k$-bit subword in $z$ that is not in $\Cyc$ is the first $k$ bits of $\sig$.
The main motivation behind the definition of the signature is the following crucial lemma:

\begin{lemma}\label{comblemma}
Let $s \in \zo^k$ for some $2\le k\le C$, and let $z$ be a string of length at most $15\sigma+1$
that is in the right form.
For $\by\sim \Del_\delta(z)$, the probability that $|\by|<|z|$ (so at least one deletion occurs)
  and $\by$ is the prefix of a string in the right form is at most $O(\delta)$.
\end{lemma}

The high-level idea is that a deletion is likely to create an additional disjoint $k$-bit subword in $\by$ that is not in $\Cyc$, unless a deletion occurs in some $O(k)$ specific places in $z$ or two deletions are $O(k)$ close to each other.
This additional subword will help us argue that $\by$ does not have the right form.

\begin{proof}
We start with the following simple claim:
\begin{claim}\label{trivialclaim}
Let $k\geq2$, and $w=(w_0,\dots,w_{2k})$ be a string of length $2k+1$ such that every $k$-bit subword of $w$ lies in $\Cyc$.
Let $w'$ be the string obtained from $w$ by deleting the middle bit $w_k$.
Then $w'$ must contain a $k$-bit subword that is not in $\Cyc$.
\end{claim}
\begin{proof}
Assume this is not the case. Then $w'_{k+i}=w'_i$ for each $i\in [0:k-1]$.
But this implies that $w_{k+i+1}=w_i$ for each $i\in [0:k-1]$.
On the other hand, $w_{k+i}=w_i$ for each $i\in [0:k]$.
These equations imply that all bits in $w$ are the same, contradicting the hypothesis that 
$k\ge 2$.
\end{proof}

Let $h$ be the location of the beginning of $\sig$ in $z$ (recall that $z$ is in the right form).
Let's start with the case when $\sig$ does not end with a $k$-bit subword that is not in $\Cyc$
  (so $\sig=\tail$ in this case and hence $\sig$ is of length $8 \sigma$).
Given $\by\sim \Del_\delta(z)$, we use $D(\by)$ to denote the set of positions of $z$ that are
  deleted to form $\by$. 
  Recall that our goal is to show that the probability that both (i) $D(\by)$ is nonempty, and 
  (ii) $\by$ is the prefix of a string in the right form, is at most $O(\delta)$.
We consider the following possibilities for how $D(\by)$ can be nonempty:
\begin{flushleft}\begin{enumerate}
\item If $D(\by)$ contains two deletions that are $O(k)$ positions away from each
  other (i.e., there exist $i,i'\in D(\by)$ with $|i-i'|=O(k)$), we can bound the probability of such a trace $\by\sim \Del_\delta(z)$ by 
$O(\sigma)\cdot O(k)\cdot \delta^2=o(\delta)$ given that the length of $z$ is $O(\sigma)$.
\item If $D(\by)$ contains a deletion $i$ such that $i\le O(k)$, $i\ge |z|-1-O(k)$, 
  or $|i-h|\le O(k)$ (i.e., either $i$ is too close to the beginning or the end of $z$ or
  too close to the beginning of $\sig$), one can bound the probability of such a trace $\by\sim \Del_\delta(z)$ by $O(k)\cdot \delta=O(\delta)$.
\item Otherwise, $D(\by)$ is nonempty and does not satisfy any of the two conditions above.
We show that $\by$ cannot be a prefix of a string in the right form. 
Let $i$ be the smallest integer in $D(\by)$.
Given that $i$ does not satisfy any of the conditions above, it follows from
  \Cref{trivialclaim} that a $k$-bit subword around the image of $z_{i-1}$ 
  in $\by$ is not in $\Cyc$, and this $k$-bit subword is disjoint from the image of the first $k$-bits of $\sig$
  which survives in $\by$.
If $i>h$, then for $\by$ to be the prefix of a string in the right form,
  the $\sig$ must start from the $y_h$ but we get a violation of the right form due to the subword around 
  the image of $z_{i-1}$.
If $i<h$, then for $\by$ to be the prefix of a right-form string,
  the $\sig$ must start from the $k$-bit subword around the image of $z_{i-1}$
  but we get a violation due to the subword starting at $y_h$ because the string $z$
  is of length at most $15\sigma+1$ so there is not enough space for $\sig$ to appear
  before $y_h$ (recall that in this case $\sig$ is of length $8 \sigma$).
\end{enumerate}\end{flushleft}
The lemma 
follows in the case that $\sig$ does not end with a $k$-bit subword that is not in $\Cyc$ by combining (1), (2), and (3) above.

Next we consider the case when $\sig$ ends with a $k$-bit subword not in $\Cyc$.
We consider the following possibilities for how $\by\sim \Del_\delta(z)$ can have $D(\by)$ be nonempty:
\begin{flushleft}\begin{enumerate}
\item If $D(\by)$ contains two deletions that are either $O(k)$ or 
  $|\sig|\pm O(k)$ positions away from each
  other (i.e.~there exist $i,i'\in D(y)$ with $|i-i'|=O(k)$ or $|\sig|\pm O(k)$), 
  one can bound the probability of such a trace $\by\sim \Del_\delta(z)$ by 
$O(\sigma)\cdot O(k)\cdot \delta^2=o(\delta).$
\item If $D(\by)$ contains a deletion $i$ such that either $i\le  O(k)$, $i\ge |z|-1-O(k)$,
  $|i-h|\le O(k)$, or $|i-h|=|\sig|\pm O(k)$, we can bound the probability of such a $\by\sim \Del_\delta(z)$ by $O(k)\cdot \delta=O(\delta)$.
  
\item Otherwise, $D(\by)$ is nonempty and does not satisfy any of the two conditions above,
  in which case we show that $\by$ cannot be the prefix of a string in the right form. 
Let $i$ be the smallest integer in $D(\by)$ so there is a $k$-bit subword in
  $\by$ around the image of $z_{i-1}$ that is not in $\Cyc$. 
We consider the following possibilities: 
\begin{enumerate}
\item If $i>h$, then both the first $k$ bits of $\sig$ (which is the leftmost such
  subword) and the $k$-bit subword
  around the image of $z_{i-1}$ are not in $\Cyc$. They are disjoint and the
  gap between them is not of length $|\sig|$.  So $\by$ cannot be the prefix of a string in the right form.
\item If $i<h$ and there is no other deletion in $D(\by)$ between $i$ and $h$,
  then similarly, the first $k$ bits of $\sig$ and the $k$-bit subword around the
  image of $z_{i-1}$ (the leftmost such subword) are not in $\Cyc$.
They are disjoint and the gap between them is not of length $|\sig|$. So $\by$ cannot
  be the prefix of a string in the right form. 
\item If $i<h$ and there is at least one deletion between $i$ and $h$,
  we denote the leftmost such deletion by $j$.
Then we have two subwords not in $\Cyc$ around $i$ and $j$, respectively.
The one around $i$ is the leftmost one and the gap between them is different from $|\sig|$. 
\end{enumerate}
\end{enumerate}\end{flushleft}
This finishes the proof of the lemma.
\end{proof}

\subsubsection{Proof of \Cref{lem:Align} when $k\ge2$} \label{sec:kistwo}




We start with the description of how $\Align$ works for the case when $k\ge 2$:

\begin{itemize}

\item
{\bf Description of $\Align$ for $k \ge 2$:}
Given a coarse estimate $\hat{\beta}$ (such that $|\hat{\beta}-(1-\delta)\dend|\le 2\sigma$), 
  $\sig\in \{0,1\}^{\le 8\sigma}$, and a trace $y$, $\Align$ checks if the restriction of $y$ to the
  interval 
$$J:=\left[\hat{\beta}-3\sigma: \hat{\beta}+12\sigma\right]$$ 
has a prefix in the right form (see \Cref{eq:rightform}), i.e., $y_J$ is of the form
\begin{equation} \label{eq:the-look-of-crowded}
y_J = w \circ \signature\circ v 
\end{equation}
so that the first $k$ bits of $\sigg$ is the leftmost 
  $k$-bit subword of $y_J$ not in $\Cyc$.
If $y_J$ is not of this form 
  $\Align$ returns $\nil$.
If $y_J$ is of this form and $\sigg$ ends at location $L\in [0:15\sigma]$ in $y_J$, 
  $\Align$  
  returns the index of the $(k-1)$-th bit of $\sig$
  (i.e., $\sig_{k-2}$ which intuitively should correspond to $x_\dend$) in $y$ with probability
  \begin{equation}
  	p_L := (1-\delta)^{15\sigma-L}, \label{eq:balance}
  \end{equation}
  and with the remaining probability returns $\nil$. Note that \[(1-\delta)^{15\sigma-L}=1-O(\delta \sigma)=1-\tilde{O}(n^{3\eps/2}),\quad\text{for all $L \in [0:15\sigma]$,}
  \] so \Align only returns \nil with probability $o(1)$ when $y_J$ is of the form \Cref{eq:the-look-of-crowded}. 

\end{itemize}

\noindent {{\bf Discussion.}  The main subtlety in the definition of \Align is the ``discounting probability'' given by \Cref{eq:balance}.  We will see in the proof of \Cref{claimhehe2} that this probability plays an important role in ensuring that the location returned by \Align (conditioned on \Align not returning $\nil$) is sufficiently close in expectation to the correct location.}


\medskip

\noindent\textbf{Correctness.} 
The first property of \Cref{lem:Align} is trivial: $\Align$ 
  always returns an integer in $J$ when it does not return $\nil$, and thus
the property follows from the assumption that $\hat{\beta}$ 
is a coarse estimate of $(1-\delta)\dend$. 
For the second property, 
  consider the 
  event over $\by\sim \Del_\delta(x)$ that the whole subword of $x$ of length $40\sigma$ 
  centered at $\dend$
  survives in $\by$ and $x_{\dend}$ falls inside $[\hat{\beta}-3\sigma:
  \hat{\beta}+3\sigma]$ 
  in $\by$. 
This happens with probability at least $1-O(\delta\sigma)\ge 1-\tilde{O}(n^{-3\eps/2})$
  (where we use the Chernoff bound and the assumption that $|\hat{\beta}-(1-\delta)\dend|\leq 2\sigma$ 
  to show that the probability of $x_{\dend}$ falling
  outside $[\hat{\beta}-3\sigma:
  \hat{\beta}+{3\sigma}]$ is $n^{-\omega(1)}$);  when it 
  happens, $\Align$ running on $\by$ returns $\pos(\by)$ 
  with probability at least $1-O(\delta\sigma)$ (\Cref{eq:balance}). As a result, we have
\begin{equation}\label{secondproperty}
 \Pr_{\by\sim \Del_\delta(x)}\big[\Align(\by)=\pos(\by)\big]\ge 1-\tilde{O}(n^{-3\eps/2}),
\end{equation}
where we use $\Align(\by)$ to denote the output of \Align on input $\by$ here and in the rest of the section for convenience\footnote{Note that  $\Align(y)$ is a random variable in general even
  for a fixed trace $y$}.
In the rest of this subsection, we establish the third property of \Cref{lem:Align}.

We start with some notation for working on the third property.
Given a trace $y$ we let $\Align'(y)$ denote the
  integer returned by $\Align$ minus $\hat{\beta}-3\sigma$ (i.e., the location
   of the $(k-1)$-th bit of $\sig$ in $y_J$), and we define $\Align'(y)$ to be $\nil$
  if $\Align$ returns $\nil$.\footnote{Just like $\Align$, $\Align'(y)$ is a random variable in general
  even for a fixed trace $y$.}
Recalling the definition of $J$, we have that  $\Align'(y)$ always lies in $[0:15\sigma]$ when it is not $\nil$.
After this shift,
  establishing the third property of \Cref{lem:Align} is equivalent to showing that
\begin{equation}\label{eq:target}
\bE_{\by\sim \Del_\delta(x)} \big[\Align'(\by)\hspace{0.04cm}\big|\hspace{0.04cm}\Align'(\by)
  \ne \nil\big]=(1-\delta)\dend -(\hat{\beta}-3\sigma)\pm o(1).
\end{equation}

\usetikzlibrary{decorations.pathreplacing}
\usetikzlibrary{math}
\tikzmath{\x = 30; \y = -5; \h= 10;}
\begin{figure}[!t]
\centering
\begin{tikzpicture}[d/.style={draw=black,thick,font=\scriptsize}, x=1mm, y=1mm, z=1mm]

  \draw[thick](0,\x)--(140,\x) node at (-6,\x)[left]{$x$};
  \draw (0,\x+2) node[anchor=west,align=left,scale=0.90] {$00100100100 \, \cdots \cdots \cdots \cdots \cdots \cdots \cdots 10{\color{red}01{\color{blue} \textbf{1}}001001001001 \cdots \cdots \cdots \cdots 00100100100{\color{blue}\textbf{0}}}10010010 \cdots $ };
  \draw[d](\h+50,\x-2)--(\h+50,\x) node[below=4]{\begin{tabular}{c} $\dend$ \\ \rotatebox[origin=c]{90}{$=$} \\ $\sig_{k-2}$\end{tabular}};
  \draw[d](\h+2,\x-2)--(\h+2,\x) node[below=2]{$\dend-6\sigma$};
  \draw[d](\h+42,\x-2)--(\h+42,\x) node[below=2]{$\dend-\sigma$};
  \draw[d](\h+26,\x-2)--(\h+26,\x) node[below=2]{$\dend-3\sigma$};
  \draw[d][decorate, decoration={brace}, yshift=2ex]  (\h+47,\x) -- node[above=0.4ex] {$\abs{\sig} \le 8\sigma$}  (\h+114,\x);

  \draw[thick](0,\y)--(140,\y) node at (-6,\y)[left]{$\by$};
  \draw (0,\y+2) node[anchor=west,align=left,scale=0.90] {$0100100 \cdots \cdots \cdots \cdots \cdots \cdot \cdots 10{\color{red}01 {\color{blue}\textbf{1}}001001001001 \cdots \cdots \cdots \cdots 00100100100{\color{blue}\textbf{0}}}10010010 \cdots $ };
  \draw[d](\h+46,\y-2)--(\h+46,\y) node[below=2]{$(1-\delta)\dend$};
  \draw[d](\h+40,\y)--(\h+40,\y+2) node[above=1]{$\hat\beta$};

  \fill[fill=red, opacity=0.2](\h+2,\x)--(\h+16,\y)--(\h+46,\x)--(\h+2,\x);
  \draw[d][decorate, decoration={brace}, yshift=2ex]  (\h+2,\x) -- node[above=0.4ex] {{\color{red}$T$}}  (\h+46,\x);

  \fill[fill=green, opacity=0.2](\h+16,\y-3)--(\h+16,\y)--(\h+110,\y)--(\h+110,\y-3)--(\h+16,\y-3);
  \draw[d][decorate, decoration={brace,mirror}, yshift=-4ex]  (\h+30,\y) -- node[below=0.4ex] {$2\sigma$}  (\h+46,\y);
  \draw[d][decorate, decoration={brace,mirror}, yshift=-4ex]  (\h+46,\y) -- node[below=0.4ex] {$2\sigma$}  (\h+62,\y);
  \draw[d](\h+16,\y-2)--(\h+16,\y) node[below=2]{$\hat\beta-3\sigma$};
  \draw[d](\h+110,\y-2)--(\h+110,\y) node[below=2]{$\hat\beta+12\sigma$};

  \fill[fill=blue, opacity=0.2](\h+30,\y)--(\h+50,\x)--(\h+62,\y)--(\h+30,\y);
\end{tikzpicture}
\caption{The interval in $x$ in the base of the red triangle is $T$; the green interval in $\by$ is $J$; and the interval in $\by$ in the base of the blue triangle is $[(1-\delta)\dend-2\sigma:(1-\delta)\dend+2\sigma]$.  The red subword in $x$ is $\sig$.  Note that the first $k$ and the last $k$ bits of $\sig$ do not belong to $\Cyc$ (here $s=001$), and we have $\dend = \sig_{k-2} = \sig_1$, the second bit of $\sig$.}
\end{figure}
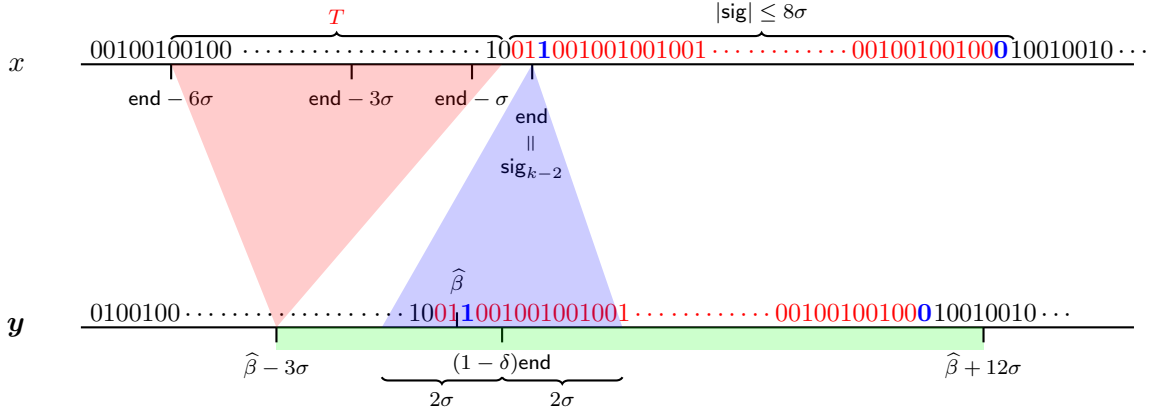

Given a trace $y$,  we let $\reach(y)$ denote the position of $x$
  that corresponds to $\smash{y_{\hat{\beta}-3\sigma-1}}$, i.e.,  the last bit
  of $x$ before the interval $J$ is reached (we set $\reach(y)=+\infty$ by default if $y$ never reaches $J$). 
By a Chernoff bound, $\by\sim \Del_\delta(x)$ has $\reach(\by)$ between
  $(\hat{\beta}-3\sigma)/(1-\delta)\pm \sigma/2$ with probability at least $1-n^{-\omega(1)}$.
Let $T$ denote the interval 
$$
T := [\dend-6\sigma:\dend-\sigma/2].$$
Recalling that  $\hat{\beta}$ is a coarse estimate of $(1-\delta)\dend$, we have that
$$
\dend-6\sigma \le (\hat{\beta}-3\sigma)/(1-\delta)-\sigma/2<
(\hat{\beta}-3\sigma)/(1-\delta)+\sigma/2 \le \dend-\sigma/2,
$$
and thus for a random $\by \sim \Del_\delta(x)$, with overwhelmingly high probability $\reach(\by)$ lies in the interval $T$. 
Let $\calT $ be the distribution of $\reach(\by)$
  (where 
  $\by\sim \Del_\delta(x)$) 
  conditioned on $\reach(\by)\in T$ (intuitively, this is a very mild conditioning which discards only a negligible fraction of outcomes of $\reach(\by) \in T$) 
  and let $\calY_t$ be the distribution of $\by\sim\Del_\delta(x)$ conditioned 
  on $\reach(\by)=t$ for some $t\in T$.

The following claim will be used to connect with our target \Cref{eq:target}:
\begin{claim}\label{simplesimple}
$\bE_{\bt\sim \calT} \big[(\hat{\beta}-3\sigma)+\dend-\bt-1\big]
=(1-\delta)\dend \pm o(1).$
\end{claim} 
\begin{proof}
  We start with the simple observation that $\bE_{\by\sim \Del_\delta(x)}[\pos(\by)+1]=(1-\delta)(\dend+1)$.
Note that $\dend+1$ is the number of bits in $x_{[0:\dend]}$ and 
  $\pos(\by)+1$ is the number of bits in $x_{[0:\dend]}$ that survive in $\by$.  
Given that $\reach(\by)\notin T$ with negligible probability, we have 
\begin{equation}\label{hehe2}
\bE_{\bt\sim \calT,\by\sim \calY_{\bt}} [\pos(\by)]=(1-\delta)\dend\pm o(1).
\end{equation}
Finally, for each $t\in T$, drawing $\by\sim \calY_{t}$ essentially corresponds to starting the deletion 
  process from $x_{t+1}$ and thus, 
  $$\bE_{\by\sim \calY_{t}}\big[\pos(\by)\big]=(\hat{\beta}-3\sigma)+(1-\delta)(\dend-t)-1.$$
Using $\dend-t=O(\sigma)$ as $t\in T$, the expectation above is 
  $(\hat{\beta}-3\sigma)+ \dend-t -1\pm o(1).$
The claim follows by combining this with \Cref{hehe2}.
\end{proof}

Below is our main technical claim:

\begin{claim}\label{claimhehe2}
We have the following two properties for each $t\in T$: 
\begin{align*}
 \bE_{\by\sim \calY_t} \big[\Align'(\by) \hspace{0.04cm}\big|\hspace{0.04cm}  \Align'(\by)\ne\nil\big]
 &=\dend-t-1\pm o(1) \quad\text{and}\\[1ex]
\Pr_{\bt\sim\calT,\by\sim\calY_{\bt}}\big[\bt=t \hspace{0.04cm}\big|\hspace{0.04cm} \Align'(\by)\ne \nil\big]
 &=(1\pm O(\delta))\cdot \Pr_{\bt\sim \calT} \big[\bt=t\big].
\end{align*}
\end{claim}

We delay the proof of \Cref{claimhehe2} and use it to prove \Cref{eq:target}. We start with
\begin{align}\label{hehe3}
 \bE_{\by\sim \Del_\delta(x)} \big[\Align'(\by)\hspace{0.04cm}\big|\hspace{0.04cm}\Align'(\by)\ne \nil\big]=
\bE_{\bt\sim \calT,\by\sim\calY_{\bt}} \big[\Align'(\by) \hspace{0.04cm}\big|\hspace{0.04cm}\Align'(\by)\ne \nil\big]\pm o(1),
\end{align}
which follows from the facts that for $\by\sim \Del_\delta(x)$,
  $\Align'(\by)\ne \nil$ with probability at least $1-o(1)$ by \Cref{secondproperty}, and 
  $\reach(\by)\notin T$ with probability at most $n^{-\omega(1)}$ {(so the additional conditioning in the expectation on the RHS has only a negligible effect, changing the expectation by $\pm o(1)$)}.

Next we rewrite the expectation on the right hand side of \Cref{hehe3} as 
\begin{align*}
\sum_{t\in T} &\Pr_{\bt\sim\calT,\by\sim\calY_{\bt}}\big[\bt=t \hspace{0.04cm}\big|\hspace{0.04cm} \Align'(\by)\ne \nil\big]
 \cdot \bE_{\by\sim \calY_t} \big[\Align'(\by) \hspace{0.04cm}\big|\hspace{0.04cm}  \Align'(\by)\ne\nil \big]\\
 &= \sum_{t\in T} \Pr_{\bt\sim\calT,\by\sim\calY_{\bt}}\big[\bt=t \hspace{0.04cm}\big|\hspace{0.04cm} \Align'(\by)\ne \nil\big]
 \cdot  \left(\dend-t -1\pm o(1)\right) \tag{first part of \Cref{claimhehe2}}\\
&=\sum_{t\in T} (1\pm O(\delta))\cdot \Pr_{\bt\sim \calT}\big[\bt=t\big]
  \cdot \left(\dend-t -1\pm o(1)\right) \tag{second part of \Cref{claimhehe2}}\\
&=\bE_{\bt\sim \calT}\big[\dend - \bt-1\big]\pm o(1)\\&=
(1-\delta)\dend-(\hat{\beta}-3\sigma)\pm o(1), \tag{by \Cref{simplesimple}}
\end{align*}
where the third equality uses the fact that for every $t \in T$, the value of $|\dend - t - 1 \pm o(1)|$ is $O(\sigma)$ and $O(\delta) \cdot O(\sigma) = o(1)$. 
This yields \Cref{eq:target} and hence \Cref{lem:Align}, 
as desired. We now turn to \Cref{claimhehe2}.

\begin{proof}[Proof of \Cref{claimhehe2}]
Given a fixed $t\in T$,
we start by showing that 
\begin{equation}\label{probcare}
\Pr_{\by\sim \calY_t}\big[\Align'(\by)\ne \nil\big]=(1-\delta)^{15\sigma+1}\pm O(\delta)
\end{equation}
independent of $t$ (up to accuracy $O(\delta)$).
The probability above can be 
  viewed as corresponding to the following probabilistic experiment: the suffix of $x$ starting at $x_{t+1}$ is fed into the deletion channel, and $\Align$ is run on the $(15\sigma+1)$-bit prefix of the resulting trace (as the 
  restriction $y_J$).
Given that  $\dend-6\sigma\le t\le \dend-\sigma/2$ for all $t \in T$,
  not only $x_\dend$ but the whole $\sig$ lies in the length-$(15\sigma+1)$ subword $z = x_{[t+1:t+15\sigma+1]}$ 
  of $x$ starting at $x_{t+1}$ (since $|\sig|\le 8\sigma$).
Let $h$ 
denote the location of the end of $\sig$ in $z$ 
(equivalently,
  $x_{t+h+1}$ 
  is the location of the end of $\sig$ in $x$).
  
Given $\by\sim \calY_t$, we write $D(\by)\subseteq [t+1:t+h+1]$ 
  to denote the set of locations
  deleted within $x_{t+1},\ldots,x_{t+h+1}$. 
We consider the following cases:
\begin{flushleft}\begin{enumerate}
\item If $D(\by)=\emptyset$ (which happens with probability 
  $(1-\delta)^{h+1}$ 
  over $\by\sim \calY_t$), then $\Align'(y)\ne \nil$ with 
  probability $p_h=(1-\delta)^{15\sigma-h}$ 
  (this is why the $\Align$ procedure includes the ``discounting probability'' in \Cref{eq:balance}). 
So this case contributes exactly $(1-\delta)^{15\sigma+1}$ to the probability in
  \Cref{probcare}.
\item 
  It follows from \Cref{comblemma} (applied on the string $z$) that 
  $D(\by)\ne \emptyset$
    and $\Align(\by)\ne \nil$ with probability 
    at most $O(\delta)$. 
\end{enumerate}\end{flushleft}
\Cref{probcare} follows by combining these two cases. The first part of the claim also follows.
To see this, let $c=O(\delta)$ be the probability (over $\by\sim \calY_t$) that 
$D(\by)\ne \emptyset$ 
  and $\Align'(\by)\ne\nil$.
Given that $\Align'(\by)$ is always $O(\sigma)$, $\dend-t-1=O(\sigma)$
  and $\delta\sigma = o_n(1)$, we have that
\begin{align*}
\bE_{\by\sim \calY_t} \big[\Align'(\by) \hspace{0.04cm}\big|\hspace{0.04cm}  \Align'(\by)\ne\nil\big]
&= \frac{(1-\delta)^{15\sigma+1}}{(1-\delta)^{15\sigma+1}+c}\cdot (\dend-t-1)
\pm \frac{c}{(1-\delta)^{15\sigma+1}+c}\cdot O(\sigma)\\
&=(1\pm O(\delta))\cdot (\dend-t-1)\pm O(\delta\sigma)\\&=\dend-t-1\pm o(1).
\end{align*}
For the second part of the claim, we have
\begin{equation}
 \Pr_{\bt\sim\calT,\by\sim\calY_{\bt}}\big[\bt=t \hspace{0.04cm}\big|\hspace{0.04cm} \Align'(\by)\ne \nil\big]
 =\frac{\Pr_{\bt\sim \calT}\big[\bt=t\big]\cdot \Pr_{\by\sim \calY_t}\big[\Align'(\by)\ne \nil\big]}{\sum_{t'\in T}\Pr_{\bt\sim \calT}\big[\bt=t'\big]\cdot \Pr_{\by\sim \calY_{t'}}\big[\Align'(\by)\ne \nil\big]}.\label{hehehe1}
\end{equation}
Using \Cref{probcare}, the sum in the denominator becomes
\[
\sum_{t'\in T} \Pr_{\bt\sim\calT} \big[\bt=t'\big]
  \cdot \left((1-\delta)^{15\sigma+1}\pm O(\delta)\right)=(1-\delta)^{15\sigma+1}\pm O(\delta).
\]
Using this and \Cref{probcare} in the numerator of \Cref{hehehe1}, and $\delta\sigma\ll 1$, we get
\begin{align*}
\Pr_{\bt\sim\calT,\by\sim\calY_{\bt}}\big[\bt=t \hspace{0.04cm}\big|\hspace{0.04cm} \Align'(\by)\ne \nil\big]
&=\frac{\Pr_{\bt\sim \calT}\big[\bt=t\big]\cdot \left((1-\delta)^{15\sigma+1}\pm O(\delta)\right)}
{(1-\delta)^{15\sigma+1}\pm O(\delta)}\\[0.5ex]&=(1\pm O(\delta))\cdot \Pr_{\bt\sim \calT} \big[\bt=t\big].
\end{align*}
This finishes the proof of \Cref{claimhehe2}, and with it the proof of \Cref{lem:Align} when $k=2.$
\end{proof}

\subsubsection{Proof of \Cref{lem:Align} when $k=1$} \label{sec:kisone}

In the case $k=1$ we have that the string $s$ is the single bit $b$ for some $b \in \zo$.
$\Align$ works in a very simple way in this case:

\begin{itemize}
\item
\noindent {\bf Description of $\Align$ for $k =1$:}
Let $J$ denote the interval 
\[
J:=\left[\hat{\beta}-3\sigma: \hat{\beta}+3\sigma\right].
\]
\Align outputs $\nil$ if the string $y_J$ contains no occurrence of $\overline{b}$; if $y_J$ does contain at least one occurrence of $\overline{b}$ then \Align outputs the location in $J$ of the first occurrence of $\overline{b}$.

\end{itemize}
{
\noindent\textbf{Correctness.} The first two properties of \Cref{lem:Align} follow from arguments that are essentially identical to the corresponding arguments for those properties in the $k \geq 2$ case. 
In fact, we get a stronger version of the second property, namely that for $\by \sim \Del_\delta(x)$, \Align returns exactly $\pos(\by)$ with probability at least $1-O(\delta)$. This is because (by a Chernoff bound)  the probability of $x_{\dend}$ (which is the first $\overline{b}$ character in $x$ after a long run of $b$'s) falling outside of interval $J$ in $\by$ is $n^{-\omega(1)}$, and the probability that $x_\dend$ is deleted in $\by$ is exactly $\delta$; as long as neither of these things happens, \Align will return exactly $\pos(\by).$ For the third property, let ${\cal D}$ denote the distribution of positions that \Align returns conditioned on its not returning $\nil$. By the discussion above, the total variation distance between distribution ${\cal D}$ and the distribution of $\pos(\by)$ for $\by \sim \Del_\delta(x)$ is at most $O(\delta)$. Since the support of ${\cal D}$ is contained in an interval of width $O(\sigma)$, it follows that the expectation of what \Align returns conditioned on its not returning $\nil$ is within $\pm O(\sigma) \cdot O(\delta) = o(1)$ of 
$\bE_{\by\sim \Del_\delta(x)}[\pos(\by)]$. Since 
$\bE_{\by\sim \Del_\delta(x)}[\pos(\by)] = (1-\delta)\dend \pm o(1)$, the third property holds by the triangle inequality.
}

\subsection{Proof of \Cref{thm:DET}}
\label{sec:proof-of-det}

\begin{proofof}{\Cref{thm:DET}}
The proof follows from the guarantees in \Cref{lem:Coarse-Estimate} and \Cref{lem:Align}, using standard concentration bounds. First, we have from \Cref{lem:Coarse-Estimate} that the output $(\hat{\beta}, t)$ of $\CE$ satisfies $|\hat{\beta}-(1-\delta)\dend| \leq 2\sigma$ and $t=\tail$ with probability $1-O(1)/n^3$. Assume that this holds for the rest of the proof.

By \Cref{lem:Align}, with probability at least $1-\tilde{O}(n^{-3\eps/2})$ $\Align$ returns an integer $\ell_i$ (and not \nil), and $\ell_i=\pos(\by^i)$,  for each $i \in [N]$. Now, the additive form of the Chernoff bound 
implies that $\ell_i = \pos(\by^i)$ for at least $0.9$ fraction of $i \in [N]$ with probability $1-\exp(-\Omega(N)) \geq 1-1/n^3$, where we choose the hidden constant in $N = O(\log n)$ to be sufficiently large.

It remains to show that $b = \dend$ with probability at least $1-1/n^3$. Recalling step~4 of \DET, let $G \subset [\gamma]$ be the set of indices $i$ for which $h_i = \Align(\hat{\beta}, t, \bz^i) \neq \nil$. Using the same argument as above, we have that $|G| \geq 0.9\gamma$ with probability $1-\exp(-\Omega(\gamma)) = 1-\exp(-\tilde{\Omega}(n^{2/3-\eps}))$. The guarantees in \Cref{lem:Align} imply that $|\E[h_i\,|\,h_i\neq \nil] - (1-\delta)\dend| \leq o(1)$ and that the random variable $h_i$ (conditioned on its not being $\nil$) always lies in an interval of width $O(\sigma)$ for all $i \in [G]$. Moreover, $\{h_i\}_{i \in G}$ are independent random variables as they are functions of independent traces.

Let $\beta = (1/|G|) \sum_{i \in G} h_i$ be the average of $h_i$ over $i \in G$. By Hoeffding's inequality and our choice of $\gamma = O(n^{2/3-\eps} \log^3 n) = O(\sigma^2 \log n)$ (with a sufficiently large hidden constant), we have
\[
\Pr[|\beta-\E[h_i\,|\,h_i\neq \nil]|\geq 0.1] \leq \exp\left(-\Omega\left(\frac{\gamma}{\sigma^2}\right)\right) \leq \exp(-\Omega(\log n)) \leq 1/n^3.
\]
By triangle inequality, $|\beta - (1-\delta)\dend| \leq 0.1$, and so $|\beta/(1-\delta) - \dend| \leq 0.2$, with probability at least $1-1/n^3$.\footnote{Here, we use the fact that $1/(1-\delta) \leq 1.01$.} Hence, the integer $b$ closest to $\beta/(1-\delta)$ is $\dend$, which implies \DET returns $\dend$ with probability at least $1-1/n^2$ (by union bound over all the failure probabilities). Finally, the runtime of \DET is dominated by the final procedure to compute $b$. Since each run of \Align on a trace takes $O(n)$ and \Align is run on $\gamma \leq n^{2/3}$ traces, \DET runs in time $O(n^{5/3})$. This concludes the proof of \Cref{thm:DET}.
\end{proofof}

\begin{flushleft}
\bibliography{allrefs}{}
\bibliographystyle{alpha}
\end{flushleft}

\appendix


\section{Preprocessing the string} \label{sec:ap-preprocess}

In this section, we prove the following simple lemma (recall \Cref{sec:preprocess}) which shows that given traces from $\Del_{\delta}(x)$, we can simulate traces from $\Del_{\delta}(z)$ where 
$z$ is of the form  $x \circ v$ (where $v$ is a known string). The additional property that $z$ has is that any desert ends ``well before the right end of $z$". More precisely, we have the following lemma. 
\begin{lemma}~\label{lem:preprocess}
There is a randomized algorithm {\tt Preprocess} which satisfies the following with probability $1- n^{-\omega(1)}$ (over its internal randomness): 
\begin{enumerate}
\item It outputs a string $v \in \{0,1\}^{n/2}$. 
\item For any unknown string $x \in \{0,1\}^n$, given access to a sample from $\Del_{\delta}(x)$, it can output a sample from $\Del_{\delta}(z)$, where $z= x \circ v$, in linear time. 
\item For any $s \in \{0,1\}^{\leq \deslen}$, the string $v$ does not have a $s$-desert. Consequently, any desert in the string $z=x \circ v$ ends at least $n/2 - (2m+1)$ bits before the end of $z$.
\end{enumerate}
\end{lemma}
\begin{proof}
The algorithm chooses the string $v$ to be a random string of length $n/2$. Items~1 and 2  easily follow from the fact that a sample from $\Del_{\delta}(z)$ can be generated by sampling $\by \sim \Del_{\delta}(x)$ and $\by' \sim \Del_{\delta}(v)$
and producing $\by \circ \by'$. 

It remains to show Item~3. Note that it suffices to show that with probability $1- n^{-\omega(1)}$, $v$ does not have a $s$-desert. It then follows that any $s$-desert in $z$ ends at least $n/2 - (2m+1)$ bits before the end of $z$.

To prove this, observe that the number of strings  of length $\leq \deslen$ is at most $2^{\deslen+1}$. Now, let $s$ be any fixed string of length $\leq \deslen$. For a randomly chosen string $v$, the probability that there is an $s$-desert is easily seen to be at most $2^{-(2m+1)} \cdot n$. This is because (i) there are at most $n$ potential starting points for the $s$-desert and (ii) with a fixed starting point, the probability of a $s$-desert is at most $2^{-(2m+1)}$. 

Thus, by a union bound, the total probability of any $s$-desert in $v$ is bounded by $2^{\deslen+1} \cdot 2^{-(2m+1)} \cdot n = n^{-\omega(1)}$. This finishes the proof. 
\end{proof}

\end{document}